\documentclass[a4paper,USenglish,cleveref, autoref, thm-restate]{lipics-v2021}
\pdfoutput=1 %uncomment to ensure pdflatex processing (mandatatory e.g. to submit to arXiv)
\hideLIPIcs  %uncomment to remove references to LIPIcs series (logo, DOI, ...), e.g. when preparing a pre-final version to be uploaded to arXiv or another public repository

 \usepackage{multicol}

\author{Xuan Kien Phung}{Département d'informatique et de recherche opérationnelle, Université de Montréal, Montréal, Québec, H3T 1J4, Canada.}{phungxuankien1@gmail.com}{https://orcid.org/0000-0002-4347-8931}{}

\author{Sylvie Hamel}{Département d'informatique et de recherche opérationnelle, Université de Montréal, Montréal, Québec, H3T 1J4, Canada.}{hamelsyl@iro.umontreal.ca}{
https://orcid.org/
0000-0002-8941-2284}{Supported by NSERC through an Individual Discovery Grant RGPIN-2016-04576}

\authorrunning{X.\,K. Phung and S. Hamel} %TODO mandatory. First: Use abbreviated first/middle names. Second (only in severe cases): Use first author plus 'et al.'

\Copyright{X. K. Phung and S. Hamel} %TODO mandatory, please use full first names. LIPIcs license is "CC-BY";  http://creativecommons.org/licenses/by/3.0/

\ccsdesc[100]{Theory of computation $\rightarrow$ Theory and algorithms for application domains; Applied computing $\rightarrow$ Law, social and behavioral sciences; Mathematics of computing $\rightarrow$ discrete mathematics} %TODO mandatory: Please choose ACM 2012 classifications from https://dl.acm.org/ccs/ccs_flat.cfm 

\keywords{Kemeny problem, Kendall-tau distance, Kemeny rule, median permutation, computational social theory} %TODO mandatory; please add comma-separated list of keywords

\category{} %optional, e.g. invited paper

\relatedversion{} %optional, e.g. full version hosted on arXiv, HAL, or other respository/website

  \theoremstyle{definition}

\title{Space reduction techniques for the $3$-wise Kemeny problem}  %TODO Please add

\titlerunning{The $3$-wise Kemeny problem} %TODO optional, please use if title is longer than one line

\ccsdesc[100]{Theory of computation $\rightarrow$ Theory and algorithms for application domains; Applied computing $\rightarrow$ Law, social and behavioral sciences; Mathematics of computing $\rightarrow$ discrete mathematics} %TODO mandatory: Please choose ACM 2012 classifications from https://dl.acm.org/ccs/ccs_flat.cfm 

\keywords{Kemeny problem, Kendall-tau distance,  Kemeny rule, Median permutation, Computational social choice, Major Order Theorems, Consensus ranking, Generalized Kemeny rule} %TODO mandatory; please add comma-separated list of keywords

\category{} %optional, e.g. invited paper

\relatedversion{} %optional, e.g. full version hosted on arXiv, HAL, or other respository/website
\nolinenumbers %uncomment to disable line numbering

\usepackage{amsmath, amssymb,  mathabx, amsfonts,enumerate}
\usepackage[all]{xy}
\usepackage{amscd,stmaryrd}
\usepackage{comment}
\usepackage{mathtools}
\usepackage{tikz} 
\usepackage{tikz-cd} 
\tikzcdset{diagrams={nodes={inner sep=7.5pt}}}
\usepackage{bbm}

\usepackage{url}
\usepackage{hyperref}

\usepackage{array}
\newcolumntype{M}[1]{>{\centering\arraybackslash}m{#1}}

\numberwithin{equation}{section}

\begin{document}

\maketitle

\begin{abstract}
Kemeny's rule is one of the most studied and well-known voting schemes with various important applications in computational social choice and biology \cite{andrieu,social,3/4,brancotte}. Recently, Kemeny's rule was  generalized via a set-wise approach by Gilbert et. al. in \cite{setwise-aaai, setwise}. This paradigm presents interesting advantages in comparison with Kemeny's rule since not only pairwise comparisons but also the discordance between the winners of subsets of three alternatives are also taken into account in the definition of the $3$-wise Kendall-tau distance between two rankings.  In spite of the NP-hardness of the 3-wise Kemeny problem which consists of computing the set of $3$-wise consensus rankings, namely rankings whose total $3$-wise Kendall-tau distance to a given voting profile is minimized, we establish in this paper several generalizations of the Major Order Theorems, as obtained by Milosz and Hamel \cite{hamel-space} for Kemeny's rule, for the $3$-wise Kemeny voting schemes to achieve a substantial search space reduction by efficiently determining in polynomial time the relative orders of pairs of alternatives. Essentially, our theorems quantify precisely the nontrivial property that if the preference for an alternative over another one in an election is strong enough, not only in the head-to-head competition but even when taking into account one or two more alternatives, then the relative order of these two alternatives in all $3$-wise  consensus rankings must be as expected. As an application, we also obtain an improvement of the Major Order Theorems for Kememy's rule. 
Moreover, we show that  the well-known 
 $3/4$-majority rule of Betzler et al. \cite{3/4} for Kemeny's rule is only valid in general for elections with no more than $5$ alternatives with respect to the $3$-wise Kemeny scheme. Several simulations and tests of our algorithms on real-world and uniform data are provided.  
% Examples are also provided to suggest that the $3$-wise Kemeny rule can be more resistant to manipulation than the classical one. 
\end{abstract}

\section{Introduction} 
In this article, an \emph{election} is a finite collection $C=\{c_1, \dots, c_n\}$ of alternatives together with a \emph{voting profile} consisting of a finite number of (not necessarily distinct) votes. A \emph{ranking} or a  \emph{vote} is simply a complete and strict total ordering $\pi \colon c_{\pi(1)}> c_{\pi(2)}> \dots >c_{\pi(n)}$ which we identify with a permutation of $\{1, 2, \dots, n\}$ also denoted by $\pi$. Here, $x>y$ means that the alternative $x$ is ranked before the alternative $y$. The space of all rankings, for fixed $n$, can be equipped with several natural distances, for example, the Kendall-tau distance which counts the number of order disagreements between pairs of elements in two permutations, namely, the bubble-sort distance between two permutations, or more generally the $k$-wise Kendall-tau distance \cite{setwise} (see Equation~\eqref{e:definition-k-wise-distance}). 
The important Kemeny problem (cf. \cite{kemeny}, \cite{kemeny-snell}, \cite{young}) consists of finding the set of $k$-wise medians of a given election, i.e., permutations whose total distance to the voting profile is minimized with respect to the $k$-wise Kendall-tau distance. In other words, a median is a ranking that agrees the most with the voting profile. 
\par 
By  taking into consideration not only pairwise comparisons but also the discordance between the winners of subsets of three alternatives, 
the $3$-wise Kemeny voting scheme seems to be more resistant to  coalitional manipulation than the classical $2$-wise  Kemeny rule: it is much more difficult for an alternative to win an election or even to simply win another specific alternative in an election under the $3$-wise Kemeny voting scheme. In fact, most of the best-known space reduction results for Kemeny's rule fail in the $3$-wise setting (see \cite[Table 1]{kien-sylvie-condorcet}), including 
the powerful Major Order Theorems discovered in \cite{hamel-space} (see Example~\ref{ex:major-order-theorem-fail-3-wise}) and the  Condorcet criterion. 
It was shown in \cite{kien-sylvie-condorcet}  that even the $2/3$ majority in every duel is not enough to guarantee that an alternative will win an election according to the $3$-wise Kemeny voting scheme. This phenomenon  is in stark contrast to  the Condorcet criterion where a Condorcet winner, namely an alternative which is  preferred by more voters than any others, must be the unique winner of the election. Nevertheless, we know that an alternative obtaining a $3/4$ majority in every duel must be the unique winner in the $3$-wise Kemeny voting scheme  \cite{kien-sylvie-condorcet}. 
\par 
In many situations, the $3$-wise Kemeny rule is also more suitable than Kemeny's rule since it puts more weight on alternatives which are more frequently ranked in top positions in the votes. Indeed,  Kemeny's rule  puts equal weight on the head-to-head competition of two alternatives $x, y$ when in a particular vote, $x$ is the winner followed by $y$ and when in another vote, $x$ and $y$ occupy the last two position in that order. However, it is reasonable to assume that typical voters only pay attention to a shortlist of their  favorite alternatives and put a rather random order for the rest of the alternatives. Such undesirable  behavior creates noises that can   alter the Kemeny median ranking while the problem can be solved effectively using the $3$-wise Kemeny voting scheme in specific instances (see Example~\ref{ex:major-order-theorem-fail-3-wise} and Appendix~\ref{s:3-wise can be more resistant}). The above limitation of the Kemeny rule leads to the notion  of weighted Kendall tau distances introduced by Kumar and Vassilvitskii \cite{kumar} as well as the notion of set-wise Kemeny distance of Gilbert et. al. \cite{setwise}.
\par 
Motivated by the above potential and interesting features of the $3$-wise Kemeny rule as well as the NP-hardness of the various Kemeny problems (see \cite{arrow}, \cite{dwork}, \cite{setwise}), 
our main goal is to formulate new quantitative results concerning the majority rules in $3$-wise Kemeny voting schemes associated with the $3$-wise Kendall-tau distance  introduced recently in  \cite{setwise}, which provide much more refined space reductions to the Kemeny problem in comparison to existing techniques in the literature.  
Our first result (cf. Theorem~\ref{t:small-3-wise-3/4}) shows that the  fundamental  
 $3/4$-majority rule of Betzler et al. \cite{3/4} for the classical Kemeny rule is also applicable for all elections with no more than $5$ alternatives with respect to the $3$-wise Kemeny scheme.  However, the $3/4$-majority rule fails in general as soon as there are at least $6$ alternatives. Note that without restriction on the number of alternatives, 
 the $5/6$-majority rule obtained in  \cite{kien-sylvie-condorcet} serves as the $3$-wise counterpart of the $3/4$-majority rule.  
 \par 
 The second and central result of the paper is the Major Order Theorem for the $3$-wise voting scheme, denoted 3MOT (Theorem~\ref{3-wise-major-order}), and its improved version, denoted Iterated 3MOT (Theorem~\ref{3-wise-major-order-iterated}), which, to the limit of our knowledge, are the most efficient space reduction techniques for the $3$-wise Kemeny rule in the literature. In essence, our Major Order Theorems show that if  the preference for an alternative $x$ over another alternative $y$ in an election is strong enough, as measured quantitatively not only in the head-to-head competition but also when taking into account the interactions with \emph{one} or \emph{two} more other alternatives, then $x$ must be ranked before $y$ in \emph{all} $3$-wise medians of the election. The corresponding algorithms not only function efficiently in polynomial time $O(n^4m)$, where $n$ is the number of alternatives and $m$ is size of the voting profile, but also drastically reduce the search space of $3$-wise medians. 
 To get an idea on the efficiency and interests of our results, let $0\leq p <1$ be the proportion of pairs of alternatives solved by 3MOT or Iterated 3MOT out of the total of $n(n-1)/2$ pairs. Then the original search space consisting of all $n!$ possible rankings would be reduced by a factor (reduction rate) at least equal to (cf. \cite[Lemma 2, Lemma 4]{linear-extension-2018}, see also Table~\ref{table:real-data})
 \begin{equation}
 \label{e:bound-restriction-space}
 \max \left(n! \left(1+ 0.5n(1-p)\right)^{-n}, \, e^{p^2n/32}  \right) \in \, \Omega(c(p)^{n}), \quad \text{ where } \,c(p)= \max\left(\frac{2}{e(1-p)}, \, e^{p^2/32} \right).
 \end{equation} 
 \par 
 % For $n=20$, the above factor is already larger than 640 billion. 
 % Hence, we find that it would be an interesting line of research to obtain a compact characterization, other than  to recognize such elections.   
 \noindent
On real-world data, especially for political elections and competitions where there exists a high similarity between the votes, our algorithms prove to be particularly useful as $p$ usually ranges from $0.6$ to $0.9$ after only a few iterations of  Iterated 3MOT  (see Table~\ref{table:real-data} and also  Appendix~\ref{appendix-examples} for some concrete examples). The performance is also very encouraging even on the hard case of uniformly generated data where, for example when $m=15$, the $3$-wise Major Order Theorems can determine the relative rankings of approximately $47\%$ pairs of alternatives on average when $n=10$ and approximately $31\%$ when $n=15$.   
 \par 
 Our results not only extend and improve  the important $2$-wise Major Order Theorem of \cite{hamel-space} (see Section~\ref{s:classical-MOT-intro} and Theorem~\ref{t:improved-2MOT-3.0}) 
 but also provide a unified approach and technique which should pave the way for the research of  more refined algorithms and quantitative properties of $k$-wise Kemeny voting schemes for $k \geq 2$. It is worth comparing our method to the space reduction method based on a $3$-wise majority digraph introduced in \cite[Theorem~3]{setwise} whose vertices are the alternatives. While we can obtain, under some assumptions on the $3$-wise digraph of an election, a set of rankings which contains  \emph{some} $3$-wise median using \cite[Theorem~3]{setwise}, our $3$-wise Major Order Theorems provide, for all    pairs of alternatives $(x,y)$, easy-to-compute and mild sufficient conditions so that $x>y$ in \emph{all} $3$-wise medians. In fact, by relaxing the conditions in our Major Order Theorems, we can determine {more} relative orderings of a pair of alternatives in \emph{some} 3-wise median (see Theorem~\ref{3-wise-major-order-with-equality}). Another major difference is that \cite[Theorem~3]{setwise} only considers the strength of the preference for $x$ over $y$ in the presence of \emph{one} other alternative $z\neq x,y$ while the $3$-wise Major Order Theorems quantify the preference for $x$ over $y$  not only in the head-to-head competition but even when taking into account  \emph{one} or \emph{two} more alternatives, which should provide a more refined space reduction method (see Example~\ref{ex:compare-gilbert}). In any case, the constraints on all $3$-wise medians found by our 3-wise Major Order Theorems should greatly accelerate other complementary space reduction methods and vice versa, nontrivial constraints obtained by other methods can serve as the inputs to boost our algorithms, especially Iterated 3MOT.  
 Table~\ref{table:summary} below
 % (cf. also \cite{kien-sylvie-condorcet} for the precise definitions of the criteria)
 summarizes our results and some well-known space reduction criteria for the classical and $3$-wise Kemeny voting schemes.  

\begin{table}[h!]
\centering
\caption{Some efficient space reduction techniques  for set-wise Kemeny voting schemes}
\label{table:summary}
% \begin{longtable}[c]{|m{4cm}|m{4cm}|m{4cm}|}
% \begin{tabulary} {\linewidth} {C C C}
\begin{tabularx}{\textwidth}{ 
  | >{\hsize=0.87\hsize \centering\arraybackslash}X 
  | >{\hsize=1.13\hsize \centering\arraybackslash}X 
  | >{\hsize=1\hsize\centering\arraybackslash}X |}
% \begin{tabular}{ |m{4cm}|m{4cm}|m{4cm}| }
%\begin{tabular}{ *{3}{c}}
%\hline
%\multicolumn{3}{|c|}{merge column} \\
\hline
\textbf{Criterion} & \textbf{Kemeny rule} & \textbf{$3$-wise Kemeny rule} \\ 
\hline  
% Monotonicity & Yes & Gilbert et al. 2022 \cite{setwise}
% \\ 
% \hline  
%  Condorcet loser criterion  & Yes & No, Phung-Hamel \cite{kien-sylvie-condorcet} \\   
%  \hline
%  Reversal symmetry & Yes & No, Phung-Hamel \cite{kien-sylvie-condorcet} \\   
%  \hline
% Majority criterion &   Phung-Hamel \cite{kien-sylvie-condorcet} &   Phung-Hamel \cite{kien-sylvie-condorcet}\\ 
% \hline
 Extended Always theorem (Pareto efficiency, unanimity)  & Phung \& Hamel \cite[Theorem~3]{kien-sylvie-condorcet} & Phung \& Hamel \cite{kien-sylvie-condorcet} (Theorem~\ref{t:3-wise-unanimity-general}) \\ \hline
$3/4$-majority rule (Section~\ref{s:section-3/4-rule})  & Betzler et al.  \cite{3/4}  & Valid only for elections of $5$ alternatives or less, Theorem~\ref{t:small-3-wise-3/4} \\\hline
Extended $s$-majority rule  &    Phung \& Hamel \cite[Section~4]{kien-sylvie-condorcet} & Valid if  $s\geq 5/6$\, \cite[Section~8]{kien-sylvie-condorcet} 
\\
\hline
Major Order Theorems  & Milosz \& Hamel  \cite{hamel-space} (Section~\ref{s:classical-MOT-intro}) 
Improved Iterated MOT (Theorem~\ref{t:improved-2MOT-3.0}) & Not valid, see Example~\ref{ex:major-order-theorem-fail-3-wise} \\ 
\hline 
$3$-wise Major Order Theorems & & Theorems~\ref{3-wise-major-order}, \ref{3-wise-major-order-iterated}, \ref{3-wise-major-order-with-equality}. 
 \\ 
\hline
%\end{tabular}
% \end{tabulary}
\end{tabularx}
% \end{longtable} 

\end{table}
\par 
To illustrate the utility of our obtained algorithms, we performed several simulations and tests on real-world and uniform data. 
% \begin{enumerate} [\rm (i)]
%     \item Python-Code-3-wise Major Order Theorem.tex:  tests on the 3-wise Major Order Theorem \ref{3-wise-major-order} over real data from PreLib.org.  
%     \item Python-Iterated-3-wise Major Order Theorem.tex: tests on the Iterated 3-wise Major Order Theorem \ref{3-wise-major-order-iterated} over real data from PreLib.org. 
% \end{enumerate}
Table~\ref{table:proportion-iterated-3mot-1} and Table~\ref{table:proportion-iterated-3mot-2} in Appendix~\ref{appendix-data} record the approximate average proportion of pairs with relative order solved by the  $3$-wise Extended Always Theorem obtained in \cite{kien-sylvie-condorcet} (see Theorem~\ref{t:3-wise-unanimity-general}) and the 3-wise Major Order Theorems
% (Theorems \ref{3-wise-major-order},   \ref{3-wise-major-order-iterated}, \ref{3-wise-major-order-with-equality}) 
over 100 000 uniformly generated instances. Several concrete  examples with real-world data (elections with a dozen and up to 250  alternatives)  taken from PREFLIB \cite{prelib} are described in Appendix~\ref{appendix-examples} (see also Table~\ref{table:real-data}). 
A direct implementation shows that for voting profiles with $n, m \leq 40$, the list of all pairs of alternatives  with relative order in all $3$-wise medians determined by 3MOT, resp. Iterated 3MOT, can be obtained in approximately less than 2.5 seconds, resp. 1 minute, with an M1 MacBook Air 16GB RAM. 
% For $n,m\geq 200$, the running time could reach a few hours but 
More optimized implementation can definitely improve the running time of the algorithms.
\par 
Finally, we explain how to apply our results and the set-wise Kemeny rule to deal with incomplete votes in Section~\ref{s:conclusion} and propose the usage of the proportion of pairs whose relative order are determined by the $3$-wise Major Order Theorems as a meaningful measure of the consensus level of the electorates (see Appendix~\ref{s:measure-consensus}). 

% All the python codes can be found at  \hyperlink{https://github.com/XKPhung}{https://github.com/XKPhung}. 

\section{Preliminaries} 

We recall in this section the classical 2-wise Major Order Theorems, the $3/4$-majority rule as well as the definition of $k$-wise medians and $k$-wise Kendall-tau distance. 

\subsection{The $k$-wise Kemeny rule} 

Let $k \geq 2$ be an integer and let $C$ be a finite set of alternatives. Let $S(C)$ be the set of all rankings of $C$. Let $\Delta^k(C) \subset 2^C$ be the collection of all subsets of $C$ which contain no more than $k$ elements. By counting the number of winner disagreements on all subsets taken from $\Delta^k(C)$, the \emph{$k$-wise Kendall-tau distance} $ d^k_{KT}(\pi, \sigma)$ between two rankings $\pi, \sigma$ of $C$ is defined by (cf. \cite{setwise}): 
\begin{equation}
\label{e:definition-k-wise-distance}
    d^k_{KT}(\pi, \sigma) = \sum_{S \in \Delta^k(C)} 
    % \left(
  \mathbbm{1}_{B(\pi, \sigma)}(S).
    % \right)
\end{equation}
 where $B(\pi, \sigma)=\{S \in \Delta^{k}(C) \colon \mathrm{top}_S(\pi) \neq  \mathrm{top}_S(\sigma)\}$ and  $\mathrm{top}_S(\pi) \in S$ denotes the highest ranked element of the restriction $\pi\vert_S$ of $\pi$ to $S$.  
% and $D_{x,y}$ denotes the Kronecker symbol which is equal to $1$ if $x=y$ and is equal to zero otherwise. 
The $k$-wise Kendall-tau distance between a ranking $\pi$ of $C$ and a collection of rankings $A$ of $C$ is $  d^k_{KT}(\pi,A) = \sum_{\sigma \in A}  d^k_{KT}(\pi,\sigma)$. 
Given a voting profile $V$ over $C$, we say that a ranking $\pi^*$ of $C$ is a \emph{median with respect to the $k$-wise Kemeny rule} or a \emph{$k$-wise median} if 
$d^k_{KT}(\pi^*,V) = \min_{\pi \in S(C)} d^k_{KT} (\pi, V)$.   
For $k=2$, we recover the Kendall-tau distance  $d^2_{KT} = d_{KT}$. It was shown in \cite{setwise} that the decision  variant of the $k$-wise Kemeny problem is NP-complete for every $k \geq 3$. 

\subsection{Transitivity lemma}

\begin{definition}
Let $x,y$ be alternatives in an election with voting profile $V$ and let $s\in [0,1]$. When $x$ is ranked before $y$ in a vote $ v \in V$, we denote $x>^v y$ or simply $x>y$ when there is no possible confusion.  
We write $x \geq_s y$, resp. $x>_sy$, if  $x>y$ in at least, resp. in more than, $s|V|$ votes. 
\end{definition}

The following simple but very useful observation serves as a weak transitivity property.

\begin{lemma}
\label{l:transitive}
    Let $x,y,z$ be three distinct  alternatives in an election with voting profile $V$. If $x\geq_s y$ and $x\geq_s z$ for some $s \in [0,1]$, then in at least $(2s-1)|V|$ votes, we have $x >y, z$, i.e., $x>y$ and $x>z$.  
\end{lemma}

\begin{proof}
Let $A, B \subset V$ be the sets of votes in which $x>y$ and $x>z$ respectively. Since $x\geq_s y$ and $x\geq_s z$, we deduce that $|A|\geq s|V|$ and $|B|\geq s |V|$. Consequently, from the set equality $|A \cup B| = |A|  + |B| - |A \cap B|$, we obtain the following estimation 
$   |A\cap B|  = |A| + |B| - |A \cup B|    \geq s|V|+ s|V| - |V| = (2s-1) |V|$. 
Since in every vote $v \in A \cap B$, we have $x>y,z$, the conclusion thus follows. 
\end{proof}

\subsection{The $3/4$-majority rule} 
\label{s:section-3/4-rule}
Following \cite{3/4}, a \emph{non-dirty pair of alternatives} in an election with respect to a threshold $s \in [0, 1]$ is a pair $(x, y)$ such that either $x$ is ranked before $y$ in at least $s|V|$ votes, or $y$ ranked before $x$ in at least $s|V|$  votes. Then we say that an alternative is \emph{non-dirty} if $(x,y)$ is a non-dirty pair with respect to the threshold $s$ for every other alternative $y \neq x$. An election with a certain voting rule satisfies the  \emph{$s$-majority rule} if for every non-dirty candidate $x$ with respect to the threshold $s$ and every other candidate $y \neq x$, the relative positions of the pair $(x,y)$ in every median is determined by the head-to-head majority rule.  
By pioneering results in \cite{3/4}, every election satisfies the $3/4$-majority rule with respect to the $2$-wise Kemeny voting scheme. We shall prove in Section~\ref{s:3/4-only-small-3-wise} that with respect to the $3$-wise Kemeny voting scheme, the $3/4$-majority rule holds for all elections with no more than $5$ alternatives but fails in general otherwise. 

\subsection{The classical $2$-wise Major Order Theorems} 
\label{s:classical-MOT-intro}
Since the scattered foundational  works of Condorcet \cite{condorcet} in 1785, of Truchon \cite{truchon-XCC} in 1990, of Kemeny \cite{kemeny} in 1959, of Young-Levenglick \cite{young-levenglick} in 1978, and of Young \cite{young} in 1988, only the last two decades have witnessed the most important discoveries in the NP-hardness complexity (Bachmeier et al.  \cite{arrow}, Dwork et al. \cite{dwork}), in heuristic and approximation algorithms (Ailon et al. \cite{ailon}, Ali-Meilă \cite{ali-meila}, Nishimura-Simjour \cite{nishimura-simjour}, Schalekamp-van Zuylen \cite{schalekamp}), in lower bounds (Cotnizer et al. \cite{conitzer}, Davenport-Kalagnanam \cite{davenport}), and in exact space reduction techniques (Betzler et al. \cite{3/4}, Blin et al. \cite{Blin}, Milosz-Hamel \cite{hamel-space}, Phung-Hamel \cite{kien-sylvie-condorcet}) of the Kemeny problem. 
Among the space reduction techniques, the Major Order Theorems obtained in \cite{hamel-space} provide some of the most efficient known algorithms. Such theorems are based on the observation that if two alternatives  are close enough in all votes they will have the tendency to be placed in their major order in any consensus ranking. 
To state the theorems, let $x,y$ be two alternatives of an election with voting profile $V$ over a set of alternatives $C$. Let $\delta_{xy}$ denotes the difference between the number of votes in which $x>y$ and respectively in which $x<y$. Let $E_{xy}$ (called the interference set) be the  multiset containing all alternatives $z$ such that $x>z>y$  in some vote where the multiplicity of $z \in E_{xy}$ is the number of votes in which $x>z>y$. 

\begin{theorem}[\textbf{MOT}, \cite{hamel-space}] 
    \label{2-wise-major-order}  
       Suppose that $x>y$ in more than half of the votes and $\delta_{xy} > |E_{yx} \setminus E_{xy}|$. Then $x$ is ranked before $y$ in every $2$-wise median of the election.   
 \end{theorem}

 \par 
 If we iterate Theorem~\ref{2-wise-major-order} by eliminating from the interference set $E_{xy}$ all alternatives which have been found to rank before (or after) both $x$ and $y$ by previous $i$ iterations to obtain the trimmed interference set $E^{(i+1)}_{xy}$, we obtain an iterated version of Theorem MOT (see Corollary~\ref{c:iterated-2-wise-major-order} for another formulation).   
 
\begin{theorem}[\textbf{Iterated MOT}, \cite{hamel-space}]  
    \label{iterated-2-wise-major-order} 
       If $x>y$ in more than half of the votes and $\delta_{xy} > |E^{(i)}_{yx} \setminus E^{(i)}_{xy}|$ for some $i\geq 0$, then $x$ is ranked before $y$ in every $2$-wise median of the election.   
 \end{theorem}
\par 
Over uniformly distributed elections with $15$ alternatives and $15$ votes,  Iterated MOT can solve on average more than 50$\%$ of possible relative orders of all pairs of alternatives \cite{hamel-space}. 

\section{The $3$-wise $3/4$-majority rule for elections with no more than 5 alternatives} 
\label{s:3/4-only-small-3-wise}

In this section, we shall prove that for the $3$-wise Kemeny voting scheme, 
the $3/4$-majority rule holds in general only 
for elections of at most $5$ alternatives. 

\begin{theorem}
\label{t:small-3-wise-3/4} 
Let $n \geq 1$ be an integer. The $3/4$-majority rule holds for all elections of $n$ alternatives with respect to the $3$-wise Kemeny voting scheme if and only if $n \leq 5$. 
\end{theorem}
\par 
\noindent 
Theorem~\ref{t:small-3-wise-3/4} results directly from Lemma~\ref{l:no-3-wise-3/4-for-8-or-more-alternatives} and Lemma~\ref{l:small-3-wise-3/4} below which prove each of the two implications in Theorem~\ref{t:small-3-wise-3/4}. 

\begin{lemma}
\label{l:no-3-wise-3/4-for-8-or-more-alternatives}
For every $n \geq 6$, the $3/4$-majority rule fails for some  elections consisting of $n$ alternatives with respect to the $3$-wise Kemeny voting scheme.  
\end{lemma}

\begin{proof}
Let $n \geq 6$ be an integer. Let $W$ denote the ranking  
$w_1 >\dots>w_{n-6}$. 
Let us consider the following voting profile $V$ consisting of $4$ votes over $n$ alternatives $x,y,z,t,u,v, w_1, \dots, w_{n-6}$   
\begin{align*}
&r_1: W>x>y>z>t>u>v \,\,\text{ (5  votes)},\quad  &r_2: W>u>y>z>x>t>v \,\,\text{ (5 votes)}, \\
&r_3: W>v>z>y>x>t>u \,\,\text{ (5 votes)},\quad  &r_4: W>t>u>x>v>y>z \,\,\text{ (4 votes)},  \\
&r_5: W>t>u>x>v>y>z \,\,\text{ (1  vote)}.\quad  & 
\end{align*} 
\par 
\noindent 
Let $\pi^*$ be a $3$-wise median of the above election. 
Then by the $3$-wise Extended Always theorem \cite{kien-sylvie-condorcet}, the first $n-8$ positions in $\pi^*$ is given by $W \colon w_1 >\dots> w_{n-6}$. Hence, we can write $\pi^* \colon W >\sigma^*$ where $\sigma^*$ is a ranking of the 
alternatives $x,y,z,t,u,v$. 
It is then clear from the definition of $d^3_{KT}$ that $\pi^*$ is a $3$-wise median if and only if $\sigma^*$ is a $3$-wise median of the following induced election $V'$: 
\begin{align*}
&r'_1: x>y>z>t>u>v \,\,\text{ (5  votes)},\quad  &r'_2: u>y>z>x>t>v \,\,\text{ (5 votes)}, \\
&r'_3: v>z>y>x>t>u \,\,\text{ (5 votes)},\quad  &r'_4: t>u>x>v>y>z \,\,\text{ (4 votes)},  \\
&r'_5: t>u>x>v>y>z \,\,\text{ (1  vote)}.\quad  & 
\end{align*} 
\par 
\noindent 
The only $3$-wise medians of $V'$ are  $ 
\sigma^*_1 \colon u>x>y>z>t>v$ and $\sigma^*_2 \colon u>y>z>x>t>v$.  
Consequently,  
$\pi^*=  \pi^*_1$ or $\pi^*=\pi^*_2$ where $ 
\pi^*_1 \colon W>u>x>y>z>t>v$ and  $\pi^*_2 \colon W>u>y>z>x>t>v$.  
Therefore, the election $V$ admits $t$ as a non-dirty alternative according to the $3/4$ threshold since $t\geq_{3/4}u$,  $t\geq_{3/4}v$, and  
$w \geq_{3/4} t$ for all $w \in \{x,y,z\} \cup\{w_i\colon 1 \leq i \leq n-6\}$. However, $t$ is ranked after $u$ in both $\pi^*_1$ and $\pi^*_2$. Thus, the $3/4$-majority rule fails for all $3$-wise medians of $V$. 
\end{proof}

\begin{lemma}
    \label{l:small-3-wise-3/4} 
The $3/4$-majority rule holds for all elections of at most $5$ alternatives with respect to the $3$-wise Kemeny voting scheme.     
\end{lemma}

For the proof of Lemma~\ref{l:small-3-wise-3/4} (and Theorem~\ref{l:second-best}), we shall need the following auxiliary result. 
\begin{lemma}
\label{l:LP-1}
    Let $C=\{x,y,z,t\}$ be a set of 4 alternatives in an election such that $z\geq_{3/4}x$, and $x\geq_{3/4}y$, and $x\geq_{3/4}t$. Then the contribution of the pair of subsets $\{x,y,t\}$ and $\{y,z,t\}$ to  
$\Delta=  d^3_{KT} (zxyt ,V) - d^3_{KT} (yzxt ,V)
    $ is at most $0$. 
\end{lemma}

\begin{proof}
  Let $m$ be the number of votes. Then by the definition of $d_{KT}^3$, the maximal contribution (divided by $m$) of the pair of subsets $\{x,y,t\}$ and $\{y,z,t\}$ to $\Delta$ is the optimal objective value of following maximization problem over $24$ variables $c_\pi$ where $\pi$ denotes a strict ranking of the alternatives $x,y,z,t$:    
\begin{align*}
&\text{maximize:}  
& \sum\limits_{\pi[y] < \pi[x], \pi[t]}  c_\pi -  & \sum\limits_{\pi[x] < \pi[y], \pi[t]} c_\pi  
+ \sum\limits_{\pi[y] < \pi[z], \pi[t]} c_\pi -
\sum\limits_{\pi[z] < \pi[y], \pi[t]} c_\pi
\\
& \text{subject to:}     
\end{align*}
\begin{align*}
    \sum\limits_{\pi}  c_\pi & =1  \\ 
 \sum\limits_{\pi[z] < \pi[x]}  c_\pi & \geq \frac{3}{4}
\\
  \sum\limits_{\pi[x] < \pi[y]}  c_\pi & \geq \frac{3}{4}\\
  \sum\limits_{\pi[x] < \pi[t]} c_\pi & \geq \frac{3}{4}\\
 c_\pi &\geq 0 \quad \text{ for all }\pi.
\end{align*}
\par 
Here, $c_\pi$ represents the number of votes in $V$ with ranking $\pi$ and $\pi[x]$, $\pi[y]$, $\pi[z]$, $\pi[t] \in \{1,2,3,4\}$ represent respectively the positions of $x,y,z,t$ in  $\pi$. 
Since the optimal objective value of the above optimization problem is $0$ (a Python code is available  at https://github.com/XKPhung), the proof is complete. 
\end{proof}

\par 

% \begin{lstlisting}[caption={Pseudo-code for the validity of the 3-wise 3/4-majority rule for 5 or less alternatives },label=LP-small-election, captionpos=t,float,abovecaptionskip=-\medskipamount]
 
% \end{lstlisting}

Generalizing the usage of linear programming in the above lemma~\ref{l:LP-1}, we obtain the following proof of Lemma~\ref{l:small-3-wise-3/4} which shows that the $3$-wise $3/4$-majority rule holds for all elections with no more than $5$ alternatives. 

\begin{proof}[Proof of Lemma~\ref{l:small-3-wise-3/4}]
Let us fix a positive integer $n \leq 5$ and a set $A=\{a_1,a_2, \dots, a_n\}$ of $n$ alternatives. Let $S_n$ denotes the set of all rankings $r$ over $A$ with the convention that $r[i] \in \{1, \dots, n\}$ denotes the position of the alternative $a_i$ in the ranking $r$. 
\par 
Then observe that the $3$-wise $3/4$-majority rule fails for some election with exactly $n$ alternatives if and only if there exist 
\begin{romanenumerate}
    \item 
a voting profile $V$ over $A$ consisting of $c_r$ votes with ranking $r$ for each $r \in S_n$; 
\item 
a non-dirty alternative $a_k \in A$ with respect to $V$ and 
the $3/4$ threshold where, up to renumbering the alternatives, we can assume that $a_i\geq_{3/4} a_k$ for all $i < k$ and $a_k \geq_{3/4} a_i$ for all $i >k$;   
\item a "bad" ranking $p \in S_n$ with respect to the $3/4$-majority rule, i.e., $(i-k)(p[i] - p[k]) <0$ for some $a_i \in A$;  
\end{romanenumerate} 
such that $p$ is at least as good as a $3$-wise consensus as any "good" ranking $q \in S_n$, i.e., $(i-k)(q[i] - q[k])>0$ for all $i \neq k$. In other words, for every good ranking $q \in S_n$, we have: 
\begin{align}
     \label{e:proof-LP-3/4-majority-rule}
 \Delta(p,q) = d^3_{KT}(p, V) - d^3_{KT}(q, V) \leq 0.
\end{align} 
\par 
Let $P_k, Q_k \subset S_n$ be respectively the set of all bad rankings and the set of all good rankings. 
By homogeneity, we can furthermore suppose that $0\leq c_r \leq 1$ for each $r \in S_n$ and $\sum_{r \in S_n} c_r=1$. Consequently, the existence of the data in (i), (ii), and (iii) satisfying 
the relation \eqref{e:proof-LP-3/4-majority-rule} is equivalence to the existence of $k \in \{1,2,\dots, n\}$,  $p \in P_k$ and the satisfiability of the following system $(\star)$ of  linear constraints on the variables $\{c_r\colon r  \in S_n\}$: 
\[
\begin{cases}
 \sum_{r \in S_n}  c_r   =1,   \quad c_r  \geq 0  &\text{for all} \quad  r \in S_n, \\ 
    \sum_{r \in S_n} \alpha_{i,k, r}\, c_r   \geq 3/4 & \text{for all} \quad i \in \{1,2,\dots, n\},\\ 
    \sum_{r \in S_n} \beta_{p, q, r} \,c_r   \leq 0 & \text{for all} \quad q \in Q_k, 
\end{cases}
\]  
\par 
\noindent
where the constants $\alpha_{i,k, r}, \beta_{p, q, r}$ are defined as follows
\begin{alphaenumerate}
    \item $\alpha_{i,k, r} = 1$ if $(i-k)(r[i]-r[k])>0$ and $\alpha_{i,k, r} = 0$ otherwise; 
    \item $\beta_{p, q, r} = d^3_{KT}(p, r) - d^3_{KT} (q,r)$ . 
\end{alphaenumerate}
\par 
\noindent
It is immediate that the conditions $  \sum_{r \in S_n} \alpha_{i,k, r}\, c_r   \geq 3/4$  for all $ i \in \{1,2,\dots, n\}$ translate the property that $a_k$ is a non-dirty alternative satisfying (ii). Moreover, we have: 
\begin{align*}
\Delta(p,q) = d^3_{KT}(p, V) - d^3_{KT}(q, V) & = \sum_{r \in S_n} c_r \, d^3_{KT}(p, r) - \sum_{r \in S_n} c_r \, d^3_{KT} (q,r)\\
& = \sum_{r \in S_n} \left(  d^3_{KT}(p, r) -  d^3_{KT} (q,r)\right)  c_r \\&=  \sum_{r \in S_n} \beta_{p, q, r} \,c_r. 
\end{align*}
\par 
\noindent
Thus, the conditions $ \sum_{r \in S_n} \beta_{p, q, r} \,c_r   \leq 0 $  for all $q \in Q_k$ altogether translate the property \eqref{e:proof-LP-3/4-majority-rule} for all $q \in Q_k$. We can check by a direct implementation with  linear programming (see, e.g., a python code at https://github.com/XKPhung)  that the system $(\star)$ does not admit a solution for any value of $k \in \{1,2,\dots, n\}$ and $p \in P_k$ as long as $n \leq 5$. 
Therefore, we can conclude that the 3-wise $3/4$-majority rule holds for all elections with no more than $5$ alternatives. 
% The corresponding pseudo-code of the above algorithm is given in Listing \ref{LP-small-election}.  
\end{proof}

\section{Some weak versions of the 3-wise 3/4-majority rule}
\label{s:weak-version-3-wise-3/4}
For elections with $6$ alternatives, we have the following weak form of the $3$-wise $3/4$-majority rule. 

\begin{theorem}
\label{t:3/4-for-6-candidates}
    Let $C$ be a set of $6$ alternatives. Suppose that in an election over $C$, we have a partition $C= A \cup \{x\} \cup B$ where $|A| \leq 2$ such that $y \geq_{3/4} x$ for all $y \in A$ and $x \geq_{3/4}z$ for all $z \in B$. Then the election satisfies the $3$-wise $3/4$-majority rule. 
\end{theorem}

\begin{proof}The proof follows immediately by using a  similar translation procedure to a problem in linear programming as described in detail in the above proof of Lemma~\ref{l:small-3-wise-3/4}. 
\end{proof}

By results in \cite{kien-sylvie-condorcet}, the $3$-wise $3/4$-majority holds for all non-dirty alternatives which win every duel by the ratio $3/4$. For non-dirty alternatives which  lose the head-to-head competition to exactly one other alternative, we have the following useful results. We drop the symbol $>$ in a ranking for simplicity. 

\begin{theorem}
\label{l:second-best}
 Let $C=\{x,z\}\cup J$ be a partition of the set of alternatives of an election $V$ such that $z\geq_{3/4}x$ and $x\geq_{3/4} y$ for all $y \in J$.  
 Then the following properties hold: 
 \begin{enumerate}[\rm (a)] 
     \item For all partitions $J=A\cup B \cup C$ where $A,B,C$ are ordered sets with $B\neq \varnothing$, we have: 
     \[ 
     d^3_{KT}(AzBxC, V) > d^3_{KT}(AzxBC, V). 
     \]
     \item For all partitions $J=  A \cup B$ where $A,B$ are ordered sets with $A \neq \varnothing$, we have: 
     \[
     d^3_{KT}(AzxB, V) > d^3_{KT}(zxAB, V).
     \]
 \item 
For all partitions $J= A \cup B \cup C$ where $A,B,C$ are ordered sets with $|B \cup C|=5$,  $|B|\leq 2$, we have: 
 \[
 d^3_{KT}(AxBzC, V) > d^3_{KT}(AzxBC, V) 
 \]
 
 \end{enumerate}
\end{theorem}

\begin{proof} 
    For (a), let us consider the rankings $ AzByxC$ and $ AzBxyC$ where $J=A\cup B \cup C$ is a partition with $A,B,C$ ordered sets. Then the only subsets that can contribute to 
    \[
     \Delta_a = d^3_{KT}(AzBxyC, V) - d^3_{KT}(AzByxC, V) 
    \]
    are $\{x,y\}$ and $\{x,y,t\}$ where $t \in C$. Since $x\geq_{3/4}y$ and $x \geq_{3/4}t$, we have $x>y, t$ in at least half of the votes (Lemma~\ref{l:transitive}). It follows that the contribution of every subset $\{x,y,t\}$ to $\Delta_a$ is at most $0$. Likewise, the contribution of the subset $\{x,y\}$ to $\Delta_a$ is at most $\frac{1}{4}|V| - \frac{3}{4}|V|=-\frac{1}{2}|V|<0$. Therefore, $\Delta_a <0$ and (a) follows by an immediate induction. 
    \par 
    For (b), let $y \in J$ and let us consider two rankings $AyzxB$ and $AzxyB$ where $J\setminus \{y\}=A\cup B  $ is a partition with $A,B$ ordered sets. Then the only subsets that can contribute to 
    \[
     \Delta_b = d^3_{KT}(AyzxB, V) - d^3_{KT}(AzxyB, V) 
    \]
    are $\{x,y\}$, $\{y,z\}$, and $\{x,y,t\}$, $\{y,z,t\}$ for all $t \in B$. Since $z\geq_{3/4}x$ and $x\geq_{3/4} y$, we have $z\geq_{1/2}y$ (Lemma~\ref{l:transitive}). Hence, the pair $(x,z)$ contributes at most $0$ to $\Delta_b$ and the pair $(x,y)$ contributes at most $\frac{1}{4}|V| - \frac{3}{4}|V|=-\frac{1}{2}|V|<0$ to $\Delta$. 
    \par 
     Lemma~\ref{l:LP-1} shows that for every $t \in B$, the pair of subsets $\{x,y,t\}$ and $\{y,z,t\}$ contributes at most $0$ to $\Delta_b$. To summarize, we have proved  that $\Delta_b <0$. It is then clear that (b) follows by an immediate induction. 
    \par 
    For (c), the only subsets that can contribute to 
    \[
     \Delta_c = d^3_{KT}(AzxBC, V) - d^3_{KT}(AxBzC, V) 
    \]
    are $\{x,z\}$, $\{y,z\}$, $\{x,z,t\}$, and $\{y,z,t\}$ for all  $y \in B$ and $t \in B\cup C$. Since $z \geq_{3/4}x$, the contribution of the pair $(x,z)$ to $\Delta_c$ is at most 
    \[
    \frac{1}{4}|V| - \frac{3}{4}|V|=-\frac{1}{2}|V|. 
    \]
    \par 
    Since $z \geq_{3/4}x$ and $x \geq_{3/4}y$ for all $y \in B$, we have $z\geq_{1/2}y$ by Lemma~\ref{l:transitive} and thus the contribution of all such pairs $(y,z)$ is at most $0$. 
    \par 
    Let $t \in B \cup C$. Under the assumption that $x\geq_{3/4} t$ and $z\geq_{3/4}x$, linear programming (see, e.g., https://github.com/XKPhung) shows that the contribution of the subset $\{x,z,t\}$ to $\Delta_c$ is at most  $-\frac{1}{4}|V|$. Hence, the total contribution of such subsets $\{x,z,t\}$ ($t \in B \cup C$) to $\Delta_c$ is at most 
    \[
   - \frac{1}{4}|V|(|B|+|C|). 
    \]
    \par 
    Similarly, under the assumption that $x\geq_{3/4}y, t$ and $z\geq_{3/4}x$, linear programming (a Python code is available at https://github.com/XKPhung) shows that for all $y \in B$ and $t \in B \cup C$, the contribution of the subset $\{y,z,t\}$  to $\Delta_c$ is at most $\frac{1}{4}|V|$. Hence, the total contribution of such subsets is at most 
    \[
    \frac{1}{4}|V|\left(\frac{|B|(|B|-1)}{2} + |B||C|\right). 
    \]
    \par 
    To summarize, we deduce the following estimation: 
    \begin{align*}
        \frac{\Delta_c}{|V|} \leq -\frac{1}{2} - \frac{1}{4}(|B|+|C|) +  \frac{1}{4}\left(\frac{|B|(|B|-1)}{2} + |B||C|\right).
    \end{align*}
\par 
Since $|B|+|C|=5$ and $|B|\leq 2$, by substituting all possible values of $(|B|,|C|)$, we deduce that $\Delta_c<0$ and the proof of (c) is complete. 
\end{proof}

 \section{$3$-wise Major Order Theorem and applications}

\subsection{Major Order Theorems fail for the $3$-wise Kemeny voting scheme} 
The example below shows that all the Major Order Theorems established in \cite{hamel-space} fail for the $3$-wise Kemeny  voting scheme. The discussion following Example~\ref{ex:major-order-theorem-fail-3-wise} also explains why the $3$-wise Kemeny voting scheme is more suitable than the Kemeny rule to avoid manipulations.

\begin{example}
    \label{ex:major-order-theorem-fail-3-wise}
 Let us consider the example taken from \cite{kien-sylvie-condorcet}. 
Let $C= \{x,y,z, t\}$ be a set of 4 alternatives and consider the voting profile $V$:  
\begin{align*}
&r_1: z>t>x>y \,\,\text{ (5 votes)},\quad  & r_2: y>t>x>z \,\,\text{ (2 votes)}, \\
&r_3: x>y>t>z \,\,\text{ (2 votes)},\quad  & r_4: t>x>y>z \,\,\text{ (2 votes)}.
\end{align*}
\par 
\noindent 
Then  $r_1$ is the unique $3$-wise median $r_1$ of the election. 
Note that the Major Order Theorems in \cite{hamel-space} do not hold for the pair of alternatives $(z,t)$ in the above election with respect to the $3$-wise Kemeny voting scheme. Indeed, $t>z$ in a majority of the votes and the multiset consisting of alternatives $c$ with $z>c>t$ in a vote is empty so that the conditions in the Major Order Theorems are satisfied. However,  we have $z>t$ in the unique $3$-wise median $r_1$ of the election. 
\end{example} 

\subsection{$3$-wise Major Order Theorem}

\begin{definition}
\label{def:number-votes-given-order}
Let $V$ be a voting profile. Let $x_1, \dots, x_k$ be some alternatives of the election. Then we denote by $n_{x_1\dots x_k}(V)$ or simply  $n_{x_1\dots x_k}$ the number of votes in $V$ in which the relative orders of $x_1, \dots, x_k$ is 
$x_1>x_2>\dots > x_k$. 
For two alternatives $x,y$, we also define $\delta_{xy} =n_{xy} -n_{yx}$. 
If $\delta_{xy}>0$, we say that the \emph{major order} and the \emph{minor order} of the pair $(x,y)$ are respectively $x>y$ and $x<y$.  
\end{definition}

\begin{definition}
For alternatives $x,y,z,t$ of an election with voting profile $V$, we define 
\begin{align*}
    P_{x,y,z} &   = 3n_{xzy} + n_{xyz}+n_{zxy}+n_{zyx} - 4n_{yzx}  - 2n_{yxz}  \\
    Q_{x,y,z} & = n_{xyz} + n_{xzy} -n_{yxz}  -n_{yzx} = - Q_{y,x,z}
    \\
    R_{y,x,z,t} &  = 2 n_{yztx} + 2n_{yzxt} +  n_{ytzx}  +  n_{ytxz}\\
    S_{y,x,z,t} & = R_{y,x,z,t} - R_{x,y,z,t}. 
\end{align*}
\end{definition}
\par
The underlying idea is that 
the quantity $S_{y,x,z,t}$ captures the preference of the alternative $y$ over the alternative $x$ in the presence of the alternatives $z$, $t$.  On the other hand, $P_{x,y,z}$ and $Q_{x,y,z}=-Q_{y,x,z}$ are certain measures of the preference of $x$ over $y$ when $z$ comes into play. Note that $S_{y,x,z,t}= -(Q_{z,y,t}+Q_{x,z,t})$ (see Lemma~\ref{appendix-3-wise-major-order-theorem-relation-s-q} below). 
Hence, the intuition is that if the preference for $x$ over $y$ in the election is strong enough, not only in the head-to-head competition but even when taking into consideration one or two more alternatives, then 
 $P_{x,y,z}$ and $Q_{x,y,z}$ tend to be positive while $S_{y,x,z,t}$ should be negative. 
 To quantify the above intuition and the interplay between the quantities  $Q_{x,y,z}$, $P_{x,y,z}$, and $S_{y,x,z,t}$, we establish the following $3$-wise Major Order Theorem in the vein of the various Major Order Theorems \cite{hamel-space} for the classical Kemeny rule.
 
 \begin{theorem}[\textbf{3MOT}]
     \label{3-wise-major-order}  
     Let $x,y$ be two alternatives of an election with voting profile $V$ over a set of alternatives $C$.  Suppose that $x>y$ is the major order and 
     \begin{enumerate} [\rm (i)] 
     \item $Q_{x,y,z}\geq 0$ for all $z \in C\setminus \{x,y\}$, 
\item 
$2 \delta_{xy}   > \sum_{z\neq x,y} \max\left(0, - P_{x,y,z} + \sum_{t \neq x,y,z} \max\left(0, S_{y,x,z,t} \right) \right)$. 
%where the sum is over all alternatives $z,t \in C\setminus\{x,y\}$. 
     \end{enumerate}
     \par 
     Then $x$ is ranked before $y$ in every $3$-wise median of the election.  
 \end{theorem}
The following observation is clear. 

\begin{lemma}
 The theoretical time complexity of the corresponding algorithm of 3MOT to determine the relative order of pairs of alternatives is $\mathcal{O}(n^4m)$ where $n$ is the number of alternatives and $m$ is the number of votes. \qed \end{lemma} 
 \par 
 For the proof of 3MOT,  we need to establish a  series of auxiliary results.  The first lemma is the following simple but useful observation. 

\begin{lemma}
    \label{l:appendix-observation-decomposition} 
     Given distinct alternatives $x,y,z,t$ in an election, the following relations hold: 
     \begin{align*}
         n_{xy} & = n_{zxy} + n_{xzy} + n_{xyz}, \\
         n_{xyz} & = n_{txyz}+n_{xtyz}+n_{xytz}+n_{xyzt}, 
     \end{align*}
     where the quantities $n_{xy}$, $n_{xyz}$, $n_{xyzt}$, etc. are defined in Definition~\ref{def:number-votes-given-order}. 
\end{lemma}

\begin{proof}
    For the relation $n_{xy} = n_{zxy} + n_{xzy} + n_{xyz}$, it suffices to observe that for each vote where $x>y$, we have exactly three mutually exclusive possibilities according to three possible relative positions of $z$ with respect to $x$ and $y$, namely, $z>x>y$ or $x>z>y$ or $x>y>z$. A similar argument shows that $n_{xyz}  = n_{txyz}+n_{xtyz}+n_{xytz}+n_{xyzt}$. 
\end{proof}

The next lemma describes the relation between the quantities $S_{y,x,z,t}$ and $Q_{z,y,t}$, $Q_{x,z,t}$. 

\begin{lemma}
 \label{appendix-3-wise-major-order-theorem-relation-s-q} 
 For all distinct alternatives $x,y,z,t$ in an election, we have \[
-  S_{y,x,z,t}=Q_{z,y,t}+Q_{x,z,t}.
\]
\end{lemma}

\begin{proof}
By direct computation, we infer from Lemma~\ref{l:appendix-observation-decomposition} and the definition of $Q_{z,y,t}$, $Q_{x,z,t}$, $R_{x,y,z,t}$,  $R_{y,x,z,t}$, and $S_{y,x,z,t}$ that: 

    \begin{align*}
        & Q_{z,y,t}+Q_{x,z,t}  \nonumber\\=  \,\,  &n_{zyt} + n_{zty} -n_{yzt}  -n_{ytz} + n_{xzt} + n_{xtz} -n_{zxt}  -n_{ztx}
  \nonumber\\ = \,\,  & n_{xzyt}+ n_{zxyt}+ n_{zyxt}+ n_{zytx} 
  \nonumber\\    
 + \,\, & n_{xzty} + n_{zxty}+ n_{ztxy}+ n_{ztyx}
 \nonumber\\ 
- \,\, & n_{xyzt}  -n_{yxzt}  -n_{yzxt}  -n_{yztx}  
 \nonumber\\  
-\,\, & n_{xytz} -n_{yxtz} -n_{ytxz} -n_{ytzx} 
 \nonumber\\ 
+ \,\, &n_{yxzt} + n_{xyzt} + n_{xzyt} + n_{xzty} 
 \nonumber\\
+ \,\, &n_{yxtz} + n_{xytz} + n_{xtyz} + n_{xtzy} 
 \nonumber\\ 
- \,\, & n_{yzxt} -n_{zyxt} -n_{zxyt} -n_{zxty} 
 \nonumber\\  - \,\, & n_{yztx}-n_{zytx}-n_{ztyx}-n_{ztxy}
  \nonumber \\ = \,\, 
 &2n_{xzyt}   +  2 n_{xzty}   -2n_{yzxt}-2n_{yztx}
   -n_{ytxz}-n_{ytzx}
   + n_{xtyz}+ n_{xtzy}
 \nonumber  \\ = \,\,   & R_{x,y,z,t} - R_{y,x,z,t}  \nonumber\\= \,\,     & -  S_{y,x,z,t}. 
    \end{align*}
    \par 
    Therefore, $ Q_{z,y,t}+Q_{x,z,t}=-S_{y,x,z,t}$ and the proof is complete. 
\end{proof}
\par 
The next lemma allows us to express the quantity $P_{x,y,z}$ as a sum of $Q_{z,y,x}$, $Q_{x,y,z}$, and  $\delta_{xz}$, $\delta_{zy}$. 
\begin{lemma}
    \label{appendix-3-wise-major-order-theorem-relation-p-q-delta}
    For all distinct alternatives $x,y,z$ in an election, we have: 
    \[
    P_{x,y,z} = \delta_{xz} +\delta_{zy} + Q_{x,y,z} + Q_{z,y,x}. 
    \]
\end{lemma}
\begin{proof}
For strict rankings, we deduce from Lemma~\ref{l:appendix-observation-decomposition} that: 
    \begin{align*}
 &  \delta_{xz} +\delta_{zy} = n_{zy}+n_{xz}  - n_{yz}- n_{zx} 
 \\ 
 = \,\, &  n_{xzy} +n_{zxy}+ n_{zyx} + n_{yxz} + n_{xyz} +n_{xzy} - n_{xyz} -n_{yxz}-n_{yzx}- n_{yzx}-n_{zyx}-n_{zxy}
 \\  =  \,\, &   2    (n_{xzy}        -n_{yzx}), 
\end{align*}
and  
\begin{align*}
    &   Q_{x,y,z}+Q_{z,y,x} \\  = & \,\,    n_{xyz} + n_{xzy} -n_{yxz}  - n_{yzx} + n_{zxy} + n_{zyx} -n_{yxz}  -n_{yzx} \\= &  \,\,
     n_{xyz} + n_{xzy}+ n_{zxy} + n_{zyx} -2n_{yxz}  -2n_{yzx}. 
\end{align*}
\par 
\noindent
Therefore, we obtain 
\begin{align}
      \delta_{xz} +\delta_{zy} +Q_{x,y,z}+Q_{z,y,x} =  P_{x,y,z}. 
\end{align}
 and the proof is thus complete.  
\end{proof}
 \par 
We can now state     
the following key technical lemma in the proof of 3MOT. The main point is that Lemma~\ref{appendix-3-wise-major-order-theorem}, using only the quantities of the form $\delta_{x,y}$, $P_{x,y,z}$, $Q_{x,y,z}$, $S_{s,y,z}$, allows us to compare simultaneously two potential consensus rankings with a given ranking where one would like to switch the relative order of a certain pair of  alternatives. 

\begin{lemma}
\label{appendix-3-wise-major-order-theorem}
    Let $\pi \colon L>y>Z>x>R$, $\sigma_1^* \colon L>Z>x>y>R$, and  $\sigma^*_2 \colon L>x>y>Z>R$ be three rankings of an election where $L,Z,R$ are linearly ordered sets of alternatives. For $i \in \{1,2\}$, let $ 
 \Delta_i= d_{KT}^3(\pi, V) - d_{KT}^3(\sigma_i^*, V)$.  
Then we have: 
 \begin{align}
 \label{e:proof-3-wise-major-order-eq-1}
 \displaystyle 
    \Delta_1 + \Delta_2 =    2 \sum_{t \in R}  Q_{x,y,t}  + 2 \delta_{xy} + \sum_{z \in Z} P_{x,y, z} - \sum_{z\in Z,\, t\in Z\cup R, \,z>^{\pi}t} S_{y,x,z,t}. 
 \end{align} 
\end{lemma}

\begin{proof}
By the definition of $\delta_{xy}$, $Q_{u,v,w}$, and the $3$-wise Kendall-tau distance $d_{KT}^3$, we find that:  
\begin{align*}
\Delta_1 = & \,\, n_{xy} - n_{yx} + \sum_{z\in Z} (n_{zy} -n_{yz}) + \sum_{z\in Z,\, t\in Z\cup R, \,z>^{\pi}t} ( n_{zyt} +n_{zty} - n_{yzt} -n_{ytz}) \nonumber  \\ & + \sum_{z \in Z} (n_{zxy} + n_{zyx} - n_{yxz} - n_{yzx}) + \sum_{t \in R} (n_{xyt} + n_{xty} - n_{yxt} - n_{ytx} )  \nonumber \\
  = & \,\,\delta_{xy} + \sum_{z\in Z} (n_{zy} -n_{yz}) + \sum_{z\in Z,\, t\in Z\cup R, \,z>^{\pi}t} Q_{z,y,t}+ \sum_{z \in Z}  Q_{z,y,x}+ \sum_{t \in R} Q_{x,y,t}  
\end{align*}
and similarly 
\begin{align*}
    \Delta_2 = & \,\, n_{xy} - n_{yx} + \sum_{z\in Z} (n_{xz} - n_{zx}) + \sum_{z\in Z,\, t\in Z\cup R, \,z>^{\pi}t} (n_{xzt} +n_{xtz} - n_{zxt} -n_{ztx}) \nonumber  \\ &  + \sum_{z \in Z} (n_{xyz} + n_{xzy} - n_{yxz} - n_{yzx}) + \sum_{t \in R} (n_{xyt} + n_{xty} - n_{yxt} - n_{ytx})  \nonumber \\
  =& \,\, \delta_{xy}  + \sum_{z\in Z} (n_{xz} - n_{zx}) + \sum_{z\in Z,\, t\in Z\cup R, \,z>^{\pi}t} Q_{x,z,t} + \sum_{z \in Z} Q_{x,y,z} + \sum_{t \in R} Q_{x,y,t}   . 
\end{align*}
\par 
\noindent 
Hence, we deduce that: 
\begin{align}
    \label{e:proof-3-wise-major-order-eq-5} 
     \Delta_1+ \Delta_2 = & \,\, 2\delta_{xy} + 2 \sum_{t \in T} Q_{x,y,t}  +  \sum_{z\in Z,\, t\in Z\cup T, \,z>^{\pi}t} (Q_{z,y,t}+Q_{x,z,t})
   \nonumber \\ & + \sum_{z\in Z} ( n_{zy}+n_{xz}  - n_{yz}- n_{zx} +Q_{x,y,z}+Q_{z,y,x}). 
\end{align}
\par 
\noindent 
On the one hand, we have $\delta_{ab}=n_{ab}-n_{ba}$ and $ P_{a,b,c} = \delta_{ac} +\delta_{cb} + Q_{a,b,c} + Q_{c,b,x}$ for all distinct alternatives $a,b,c$  by Lemma~\ref{appendix-3-wise-major-order-theorem-relation-p-q-delta}. Therefore,   
\begin{align}
\label{e:proof-3-wise-major-order-eq-6}
     \sum_{z\in Z} ( n_{zy}+n_{xz}  - n_{yz}- n_{zx} +Q_{x,y,z}+Q_{z,y,x})= \sum_{z\in Z} P_{x,y,z}. 
\end{align}
\par 
\noindent 
On the other hand, since 
$- S_{b,a,c,d} = Q_{c,b,d} + Q_{a,c,d}$ for all  distinct alternatives $a,b,c,d$ by  
 Lemma~\ref{appendix-3-wise-major-order-theorem-relation-s-q}, we deduce that:   

\begin{align}
\label{e:proof-3-wise-major-order-eq-7}
 \sum_{z\in Z,\, t\in Z\cup R, \,z>^{\pi}t} (Q_{z,y,t}+Q_{x,z,t}) = - \sum_{z\in Z,\, t\in Z\cup R, \,z>^{\pi}t}  S_{y,x,z,t}. 
\end{align}
\par 
\noindent
By combining the equalities \eqref{e:proof-3-wise-major-order-eq-5}, \eqref{e:proof-3-wise-major-order-eq-6}, \eqref{e:proof-3-wise-major-order-eq-7}, we obtain the desired relation \eqref{e:proof-3-wise-major-order-eq-1}.  
\end{proof}

We can finally give the proof of 3MOT. 

\begin{proof}[Proof of Theorem~\ref{3-wise-major-order}]
Suppose on the contrary that $\pi \colon L>y>Z>x>R$ is a $3$-wise median of the election where $L,Z,R$ are ordered sets of alternatives. We will show that for the rankings $\sigma_1^* \colon L>Z>x>y>R$ and  $\sigma^*_2 \colon L>x>y>Z>R$, there exists $i \in \{1,2\}$ such that  
 \[
 \Delta_i= d_{KT}^3(\pi, V) - d_{KT}^3(\sigma_i^*, V) >0.
 \]
\par 
\noindent 
 Indeed, by Lemma~\ref{appendix-3-wise-major-order-theorem} on the differences  $\Delta_1$ and $\Delta_2$, we find that:
 \begin{align}
 \label{e:3-wise-major-order-eq-1}
 \displaystyle 
    \Delta_1 + \Delta_2 =  2 \sum_{t \in R}  Q_{x,y,t}  + 2 \delta_{xy} + \sum_{z \in Z} P_{x,y, z}  - \sum_{z\in Z,\, t\in Z\cup R, \,z>^{\pi}t} S_{y,x,z,t}. 
 \end{align} 
 \par 
 \noindent 
 We deduce that $\Delta_1+ \Delta_2>0$ by conditions (i) and (ii). Note that the inequality is strict by condition (ii). Therefore, there exists $i \in \{1,2\}$ such that $\Delta_i>0$. In other words, $\sigma^*_i$ would then be a strictly better $3$-wise Kemeny consensus than $\pi$, which contradicts the hypothesis that $\pi$ is a $3$-wise median of the election. Therefore, the alternative 
  $x$ must be ranked before $y$ in every $3$-wise median of the election. 
 % The proof is thus complete. 
 %If there are only two alternatives $x$, $y$ in the election then the $3$-wise Kemeny voting scheme coincides with the Kemeny scheme and thus the Condorcet winner criterion applies. Hence, in this case, $x$ is ranked before $y$ in every $3$-wise median of the election. 
 \end{proof}

\begin{example}
\label{ex:3-wise-major-order-1}
In Example~\ref{ex:major-order-theorem-fail-3-wise} where the Major Order Theorems \cite{hamel-space} fail with respect to the $3$-wise Kemeny voting scheme, the $3$-wise Major Order Theorem \ref{3-wise-major-order} tells us that 
in every $3$-wise median, we have $t>x$ and $x>y$. Therefore, $t>x>y$ in every $3$-wise median and we only need to determine the position of the alternative $z$, i.e., to check which of the following 4 rankings is the $3$-wise median: 
\begin{align*}
&r_1: z>t>x>y,\quad \quad &r_2: t>z>x>y,\\
&r_3: t>x>z>y,\quad \quad &r_4: t>x>y>z. 
\end{align*}
 The research space is thus reduced by $4!/4=6$ times. Note that $r_1$ is the unique $3$-wise median. 
\end{example}

\begin{example}
\label{ex:not-3-wise-major-order-1}
    In Example~\ref{ex:major-order-theorem-fail-3-wise}, we have 
    $\delta_{tz}=1$, $P_{t,z,x}=4$, and $P_{t,z,y}= 0$ but $Q_{t,z,x} = -1$, $Q_{t,z,y}=-3$,  and moreover $S_{z,t,y,x}= 2$, $S_{z,t,x,y}= 4$. Thus, the preference for $t$ over $z$ is not strong enough to guarantee that $t$ is ranked before $z$ in a $3$-wise median: all the conditions (i), (ii) in the $3$-wise Major Order Theorem (Theorem~\ref{3-wise-major-order}) fail for the pair $(t,z)$ in the majority order.  
\end{example}

\subsection{Applications on consecutive alternatives in a median} 
 As an immediate consequence of the above proof of Theorem~\ref{3-wise-major-order}, we obtain the following criteria for the ranking of two consecutive alternatives in a $3$-wise median which extends the corresponding criteria for the classical Kemeny rule \cite{Blin}. 

\begin{theorem}
\label{e:3-wise-major-order-consequence-1}
Let $x,y$ be two consecutive alternatives in a $3$-wise median ranking $\pi^*$ of an election. Suppose that 
$\delta_{xy} > \sum_{z \neq x,y} \max(0,Q_{y,x,z})$. Then $x$ must be ranked before $y$ in $\pi^*$. 
\end{theorem}

 \begin{proof}
 Suppose on the contrary that $\pi^* \colon L>y>x>R$ where $L,R$ are ordered sets of alternatives. Let  $\sigma^* \colon L>x>y>R$, we denote 
 \[
 \Delta= d_{KT}^3(\pi^*, V) - d_{KT}^3(\sigma^*, V).
 \]
 \par 
 \noindent 
By the relation \eqref{e:3-wise-major-order-eq-1} in the proof of Theorem~\ref{3-wise-major-order} (cf. Lemma~\ref{appendix-3-wise-major-order-theorem}), we deduce that 
\begin{align*}
 \label{c:3-wise-major-order-eq-2}
 \displaystyle 
  2  \Delta & =  2\delta_{xy} + 2 \sum_{t \in R}  Q_{x,y,t}   = 2\delta_{xy}  -2 \sum_{t \in R}  Q_{y,x,t}   \geq 2\delta_{xy}  -2 \sum_{z \neq x,y} \max(0, Q_{y,x,z}) 
 \end{align*}
 which is strictly positive by the hypothesis. Hence, $\sigma^*$ must be a strictly  better $3$-wise Kemeny consensus than $\pi^*$, which is a contradiction. The theorem thus follows. 
 \end{proof}

\begin{corollary}
\label{e:3-wise-major-order-consequence-2}
Let $x,y$ be two consecutive alternatives in a $3$-wise median  $\pi^*$ of an election. If $x>y$ is the major order and  
$Q_{x,y,z} \geq 0$ for every alternative $z \neq x,y$, then $x$ must be ranked before $y$ in $\pi^*$.  
\end{corollary}

\begin{proof}
Since for every alternative $z \neq x,y$, we have $Q_{x,y,z} \geq 0$ thus $Q_{y,x,z} \leq 0$ and 
$\max(0, Q_{y,x,z})=0$. Hence, as $x>y$ is the major order, we deduce that 
$$
\delta_{xy} > 0 =  \sum_{z \neq x,y} \max(0,Q_{y,x,z})$$ and the result follows from Theorem~\ref{e:3-wise-major-order-consequence-1}. 
\end{proof}

\section{Iterated $3$-wise Major Order Theorem} 
Developing the ideas of the classical Iterated MOT (Theorem ~\ref{iterated-2-wise-major-order}), 
we can strengthen even further  the $3$-wise Major Order Theorem \ref{3-wise-major-order} based on the following two observations. 
The first one is the transitivity of preferences in a single ranking as used in the proof of Theorem ~\ref{iterated-2-wise-major-order}. 
Suppose that $(x,y)$ and $(y,z)$ are  ordered pairs of alternatives found by the $3$-wise Major Order Theorem \ref{3-wise-major-order} such that $x>y$ and $y>z$ in every $3$-wise median. Then by transitivity, we also have $x>z$ in every $3$-wise median. 
\par 
The second observation is similar but slightly more sophisticated. Suppose that we know $z>x$ and $z>y$ in every $3$-wise median. Then it is clear that a $3$-wise median of the election must be of the form $L>x>Z>y>R$ or $L>y>Z>x>R$ where $L,Z,R$ are ordered set of alternatives such that $z \in L$. Then we can completely ignore such an alternative $z$ in conditions (i) and (ii) in the $3$-wise Major Order Theorem \ref{3-wise-major-order}. Similarly,  if we know that   $x>z$ and $y>z$ in every $3$-wise median, then we can ignore the alternative $z$ in condition (ii) in the outermost  summation. In fact, Theorem~\ref{3-wise-major-order-iterated} shows that we can even further relax conditions (i) and (ii) by a more careful analysis. 
\par 
We can obviously iterate in parallel the above two processes corresponding to the two observations to obtain larger and larger sets consisting of constraints determined by  ordered pairs of alternatives until these sets stabilize. 
Hence, we obtain the following iterated version of the $3$-wise Major Order Theorem. 

\begin{definition}
    \label{d:3-wise-major-order-iterated-1} Let $V$ be an election. Let $U$ be an anti-symmetric set consisting of ordered pairs of distinct alternatives $(x,y)$ in the election with no two pairs violating the  transitivity property, i.e., 
    \begin{enumerate} [\rm (a)]
        \item 
    $(y,x) \notin U$ whenever  $(x,y) \in U$,  
    \item
    $(z,x) \notin U$ whenever  $(x,y), (y,z) \in U$.   
        \end{enumerate} 
    \par 
    \noindent  We denote by $\overline{U} \supset U$ the \emph{transitive closure} of $U$:  
    if $(x,y), (y,z) \in U$ then $(x,z) \in \overline{U}$. We define also 
    \begin{align*} 
    L_{x}(U) & =\{z \colon (z,x)  \in \overline{U} \}, 
 &   L_{x,y}(U) & =\{z \colon (z,x), (z,y) \in \overline{U} \},&\\
 R_{x}(U) & =\{z \colon (x,z) \in \overline{U} \}, & R_{x,y}(U) & =\{z \colon (x,z), (y,z) \in \overline{U} \}. & 
 \end{align*} 
\end{definition}

\begin{theorem}[\textbf{Iterated 3MOT}]
    \label{3-wise-major-order-iterated}  Let $V$ be an election over a set of alternatives $C$. Let $U_{-1}=\varnothing$ and 
for $k \geq 0$, we define inductively $U_k$ as the set of ordered pairs $(x,y)$ of distinct alternatives such that
\par 
\noindent 
     \begin{enumerate} [\rm (i)] 
     \item $Q_{x,y,z}\geq 0$ for all $z \in C\setminus \left(L_{y}\left(\overline{U}_{k-1}\right)\cup \{x,y\}\right)$, 
\item 
$2 \delta_{xy}   > \sum_{z\in X_k} \max\left(0, - P_{x,y,z} + \sum_{t \in Y_k \setminus \{z\}} \max\left(0, S_{y,x,z,t} \right) \right)$ 
     \end{enumerate}
     \par 
     \noindent 
where $X_k= C\setminus \left(L_{y}\left(\overline{U}_{k-1}\right) \cup R_{x}\left(\overline{U}_{k-1}\right) \cup \{x,y\} \right)$ and $Y_k= C\setminus\left(L_{y}\left(\overline{U}_{k-1}\right) \cup \{x,y\} \right)$. 
     Then for every integer $k \geq -1$, we have ${U}_k \subset U_{k+1}$.  Moreover, if $(x,y) \in \overline{U}_{k}$ for  some   $k \geq 0$, the alternative $x$ is ranked before $y$ in every $3$-wise median of the election. 
\end{theorem}

\begin{proof}
Since $U_{-1}=\varnothing$, we have  $U_{-1} \subset U_0$ trivially. Note that $U_0$ is exactly the set consisting of ordered pairs $(x,y)$ of distinct alternatives satisfying Theorem~\ref{3-wise-major-order}. The first assertion $U_k \subset U_{k+1}$ for all $k \geq -1$ is then clear by an  immediate induction on $k$ using the definition of conditions (i) and (ii) and the obvious properties $\overline{A} \subset \overline{B}$, $L_{y}(A) \subset L_{y}(B)$, and $R_{x}(A) \subset R_{x}(B)$ for all $A \subset B$. 
\par 
Consider the second assertion that  $x$ is ranked before $y$ in every $3$-wise median of the election whenever $(x,y) \in \overline{U}_{k}$ for  some   $k \geq 0$. The case  $k=0$ results directly from Theorem~\ref{3-wise-major-order}. 
Assume that the second assertion  of Theorem~\ref{3-wise-major-order-iterated} holds for some $k \geq 0$ and let $(x,y) \in U_{k+1}$. Suppose on the contrary that  $\pi \colon L>y>Z>x>R$ is a $3$-wise median of the election $V$ where $L,Z,R$ are ordered sets of alternatives. 
Consider the rankings $\sigma_1^* \colon L>Z>x>y>R$ and  $\sigma^*_2 \colon L>x>y>Z>R$ as in the proof of Theorem~\ref{3-wise-major-order}. Then for 
 \[
 \Delta_i= d_{KT}^3(\pi, V) - d_{KT}^3(\sigma_i^*, V), \quad i\in \{1,2\}, 
 \]
 we also have (cf.~Lemma~\ref{appendix-3-wise-major-order-theorem}): 
 \begin{align}
 \label{e:3-wise-major-order-eq-1-iterated}
 \displaystyle 
    \Delta_1 + \Delta_2 =  2 \sum_{t \in R}  Q_{x,y,t}  + 2 \delta_{xy} + \sum_{z \in Z} P_{x,y, z}  - \sum_{z\in Z,\, t\in Z\cup R, \,z>^{\pi}t} S_{y,x,z,t}. 
 \end{align} 
 \par 
\noindent 
By our induction hypothesis, we have 
$L_{y}\left(\overline{U}_{k}\right) \subset L$ and $R_{x}\left(\overline{U}_{k}\right) \subset R$.  
Consequently, we find that 
\begin{align}
\label{e:3-wise-major-order-proof-2-iterated}
Z &\subset X_k= C\setminus \left(L_{y}\left(\overline{U}_{k}\right) \cup R_{x}\left(\overline{U}_{k}\right) \cup \{x,y\} \right) \text{ and}\\ 
R & \subset Y_k \setminus \{z\} = C\setminus\left(L_{y}\left(\overline{U}_{k}\right) \cup \{x,y,z\} \right). \nonumber 
\end{align}
\par 
\noindent
Therefore, we infer from   \eqref{e:3-wise-major-order-proof-2-iterated}, \eqref{e:3-wise-major-order-eq-1-iterated} and conditions (i) and (ii) satisfied by $(x,y)$ that 
\begin{align*}
    \Delta_1 + \Delta_2&=  2 \sum_{t \in R}  Q_{x,y,t}  + 2 \delta_{xy} + \sum_{z \in Z} P_{x,y, z}  - \sum_{z\in Z,\, t\in Z\cup R, \,z>^{\pi}t} S_{y,x,z,t} >0. 
\end{align*}
\par 
\noindent
We deduce that either $\Delta_1 >0$ or $\Delta_2>0$. Thus, either $\sigma_1^*$ or $\sigma_2^*$ is a strictly better $3$-wise consensus than $\pi$, which is a contradiction to the choice of $\pi$. Hence, $x$ must be ranked before $y$ in every $3$-wise median whenever $(x,y) \in U_{k+1}$. 
As every $3$-wise median of $V$ is a linear, strict, and complete ranking of the set $C$ of alternatives, it follows from the transitivity that $x$ must be ranked before $y$ in every $3$-wise median whenever $(x,y) \in \overline{U}_{k+1}$.
The proof is complete. 
\end{proof}
\par 
In practice, Theorem~\ref{3-wise-major-order-iterated} requires 3 to 6 iterations. 
We now compare Theorem~\ref{3-wise-major-order-iterated} with the 3-wise majority digraph method of Gilbert et al.  \cite[Theorem~3]{setwise} on the voting profile taken from \cite[Example~5]{setwise}. 

\begin{example}
\label{ex:compare-gilbert}
Consider the following voting profile $V$ over a set of 6 alternatives: 
\begin{align*}
&r_1: c_1>c_2>c_4>c_3>c_5>c_6 \,\,\text{ (4  votes)},\quad  &r_3: c_6>c_1>c_2>c_4>c_3>c_5 \,\,\text{ (1 vote)},\\
&r_2: c_1>c_3>c_2>c_4>c_5>c_6 \,\,\text{ (4 votes)},\quad  &r_4: c_6>c_1>c_4>c_3>c_2>c_5 \,\,\text{ (1 vote)}. 
\end{align*} 
By \cite[Theorem~3]{setwise}, there is \emph{some} 3-wise median of $V$ among the 4 rankings $c_1>c_2>c_3>c_4>c_5>c_6$, $c_1>c_2>c_3>c_4>c_6>c_5$, $c_1>c_2>c_4>c_3>c_5>c_6$, and $c_1>c_2>c_4>c_3>c_6>c_5$. 
On the other hand, Iterated 3MOT (Theorem~\ref{3-wise-major-order-iterated})  obtains (after 4 iterations) the following constraints $c_1 > c_2 > c_4 > c_5 > c_6$ and $c_1>c_3>c_5$ in \emph{all} 3-wise medians of $V$. 
Consequently, Iterated 3MOT implies that \emph{every} 3-wise median must be one of the following 3 rankings 
$c_1 > c_3 > c_2 > c_4 > c_5 > c_6$, $c_1 > c_2 > c_3 > c_4 > c_5 > c_6$, and $c_1 > c_2 > c_4 > c_3 > c_5 > c_6$. 
Note that the only 3-wise median of $V$ is $c_1>c_2>c_4>c_3>c_5>c_6$. 
\end{example}
\par 

\section{Improved Iterated MOT}
As a byproduct of our technique, we also obtain the following net improvement of Iterated MOT of Milosz and Hamel (Theorem~\ref{iterated-2-wise-major-order}) for classical Kemeny voting scheme. 

\begin{theorem}
\label{t:improved-2MOT-3.0} 
Let $V$ be an election over a set of alternatives $C$. 
% Let $U_0$ be the set consisting of ordered pairs $(x,y)$ of distinct alternatives satisfying Theorem~\ref{3-wise-major-order}.
Let $W_{-1}=\varnothing$ and 
for every $k \geq 0$, we define inductively $W_k$ as the set of ordered pairs $(x,y)$ of distinct alternatives such that
$ 2 \delta_{xy}  > \sum_{z\in Z_k}  \max\left(0, \delta_{yz} + \delta_{zx} \right)$ 
where $Z_k= C\setminus \left( L_{y}\left(\overline{W}_{k-1}\right) \cup R_{x}\left(\overline{W}_{k-1}\right) \cup \{x,y\}\right)$. 
Suppose that  $(x,y) \in \overline{W}_{k}$ for  some   $k \geq 0$. Then $x$ is ranked before $y$ in every $2$-wise median of the election. 
\end{theorem}
\par 
Indeed, in Theorem~\ref{t:improved-2MOT-3.0}, if we replace $Z_k$ by a larger set $$Z'_k=C\setminus \left( L_{x,y}\left({W}_{k-1}\right) \cup R_{x,y}\left({W}_{k-1}\right) \cup \{x,y\}\right)$$ for all $k \geq 0$ then we obtain as a consequence  Iterated MOT (see Corollary~\ref{c:iterated-2-wise-major-order} and Lemma~\ref{l:iterated-2-wise-major-order}). Therefore, in comparison with  Iterated MOT, Theorem~\ref{t:improved-2MOT-3.0} will either find more constraints on the relative orders of pairs of alternatives in all 2-wise medians  or reduce the number of required iterations otherwise. 

\begin{proof}[Proof of Theorem~\ref{t:improved-2MOT-3.0}]  
We proceed by induction on $k \geq -1$  as in the proof of Theorem~\ref{3-wise-major-order-iterated}. The case $k=-1$ is trivial since $W_{-1}=\varnothing$. Assume that the conclusion holds for some integer $k \geq -1$ and let $(x,y) \in W_{k+1}$. Suppose on the contrary that  $\pi \colon L>y>Z>x>R$ is a $2$-wise median of the election where $L,Z,R$ are ordered sets of alternatives. 
Let  $\sigma_1^* \colon L>Z>x>y>R$ and  $\sigma^*_2 \colon L>x>y>Z>R$.  Let $
 \Delta_i= d_{KT}^2(\pi, V) - d_{KT}^2(\sigma_i^*, V)$ where $i\in \{1,2\}$. We find by a direct  inspection that:  
 \begin{align*}
 % \label{e:2-wise-major-order-eq-1-iterated}
 \displaystyle 
    \Delta_1 + \Delta_2 =  \left( \delta_{xy} -  \sum_{z \in Z} \delta_{yz} \right) + \left( \delta_{xy} -  \sum_{z \in Z} \delta_{zx} \right)= 2\delta_{xy} -    \sum_{z \in Z} \left( \delta_{yz} + \delta_{zx} \right). 
 \end{align*} 
 \par 
\noindent 
By our induction hypothesis, we have 
$L_{y}\left(\overline{W}_{k}\right) \subset L$ and $R_{x}\left(\overline{W}_k\right) \subset R$.  
Therefore, 
\begin{align*}
\label{e:2-wise-major-order-proof-2-iterated}
Z \subset Z_k = C\setminus \left(L_{y}\left(\overline{W}_k\right) \cup R_{x}\left(\overline{W}_k\right) \cup \{x,y\} \right). 
\end{align*}
\par 
\noindent
Consequently, since $(x,y) \in W_{k+1}$,  we deduce that 
\begin{align*}
    \Delta_1 + \Delta_2&=   2\delta_{xy} -    \sum_{z \in Z} \left( \delta_{yz} + \delta_{zx} \right) \geq 2\delta_{xy} -    \sum_{z \in Z_k} \max \left(0,  \delta_{yz} + \delta_{zx} \right) >0.  
\end{align*}
\par 
\noindent
Hence, $\Delta_i >0$ for some $i \in \{1,2\}$ and $\sigma_i^*$ is a strictly better $2$-wise consensus than $\pi$, which contradicts the choice of $\pi$. Thus, $x$ is ranked before $y$ in every $2$-wise median whenever $(x,y) \in W_{k+1}$. By transitivity, $x$ must be ranked before $y$ in every $2$-wise median whenever $(x,y) \in \overline{W}_{k+1}$. 
The proof is complete. 
\end{proof}
\par 
As a direct consequence of Theorem~\ref{t:improved-2MOT-3.0}, we obtain the following weaker result which turns out to be equivalent to Iterated MOT (Theorem~\ref{iterated-2-wise-major-order}).   
\begin{corollary}
\label{c:iterated-2-wise-major-order}
Let $V$ be an election over a set of alternatives $C$. 
% Let $U_0$ be the set consisting of ordered pairs $(x,y)$ of distinct alternatives satisfying Theorem~\ref{3-wise-major-order}.
Let $W'_{-1}=\varnothing$ and 
for every $k \geq 0$, we define inductively $W'_k$ as the set of ordered pairs $(x,y)$ of distinct alternatives such that
$ 2 \delta_{xy}  > \sum_{z\in Z'_k}  \max\left(0, \delta_{yz} + \delta_{zx} \right)$ 
where $Z'_k=C\setminus \left( L_{x,y}\left({W}'_{k-1}\right) \cup R_{x,y}\left({W}'_{k-1}\right) \cup \{x,y\}\right)$. 
Suppose that  $(x,y) \in \overline{W'_k}$ for  some   $k \geq 0$. Then $x$ is ranked before $y$ in every $2$-wise median of the election.
\end{corollary}

\begin{proof}
By an immediate induction on $k$, it is clear that $W'_k \subset W_k$ and $Z_k \subset Z'_k$ for all integers $k \geq -1$. Indeed, the case $k=-1$ is trivial.  For the induction step, it suffices to use the the induction hypothesis $W'_k \subset W_k$ and observe that (cf. Definition~\ref{d:3-wise-major-order-iterated-1})
\begin{align*}
L_{x,y}(W'_{k}) \subset L_{x,y}(W_k) \subset L_{y}(W_k) \subset L_y(\overline{W}_k), \quad \quad 
R_{x,y}(W'_{k}) \subset R_{x,y}(W_k) \subset R_{x}(W_k) \subset R_{x}(\overline{W}_k),
\end{align*}
which subsequently imply $Z_{k+1} \subset Z'_{k+1}$ and thus $W'_{k+1} \subset W_{k+1}$.  
\end{proof}

\begin{lemma}
\label{l:iterated-2-wise-major-order}
Theorem~\ref{iterated-2-wise-major-order} is equivalent to Corollary~\ref{c:iterated-2-wise-major-order}. In particular, Theorem~\ref{iterated-2-wise-major-order} is a consequence of Theorem~\ref{t:improved-2MOT-3.0}.  
\end{lemma}

\begin{proof}
Let $V$ be an election over a finite set of alternatives $C$ and let the notations be as in Theorem~\ref{iterated-2-wise-major-order} and Corollary~\ref{c:iterated-2-wise-major-order}. Then for all distinct alternatives $x,y$, it is not hard to see that 
$|E_{xy}| = \sum_{z \neq x,y} n_{xzy}$ and  
$|E_{yx} \setminus E_{xy}|= \sum_{z \neq x,y}\max\left(0, n_{yzx} - n_{xzy}   \right)$. 
Similarly, let $Z'_k=C\setminus \left( L_{x,y}\left({W}_{k-1}\right) \cup R_{x,y}\left({W}_{k-1}\right) \cup \{x,y\}\right)$ then for all $k \geq 0$, we have $|E^{(k)}_{xy}| = \sum_{z \in Z_k'} n_{xzy}$ and  
$|E^{(k)}_{xy} \setminus E^{(k)}_{yx}| = \sum_{z \in Z'_{k}} \max\left(0, n_{yzx} - n_{xzy}\right)$. Moreover, 
we infer from Lemma~\ref{l:appendix-observation-decomposition} that  
\begin{align*}
    & \delta_{yz} + \delta_{zx} \\ = \,\,& n_{yz} - n_{zy} + n_{zx} - n_{xz} \\ 
     = \,\,& (n_{xyz} + n_{yxz} + n_{yzx}) - (n_{xzy} + n_{zxy} + n_{zyx}) 
    + (n_{yzx} + n_{zyx} + n_{zxy}) - (n_{yxz} + n_{xyz} + n_{xzy})\\
     =\, \,& 2(n_{yzx} - n_{xzy}). 
\end{align*}
\noindent
Therefore, for every $k \geq 0$, the condition $\delta_{xy} > |E^{(k)}_{xy} \setminus E^{(k)}_{yx}| $ in Theorem~\ref{iterated-2-wise-major-order} is equivalent to the condition $2\delta_{xy} > \sum_{z\in Z'_k}  \max\left(0, \delta_{yz} + \delta_{zx} \right)$ in Corollary~\ref{c:iterated-2-wise-major-order}. The proof is thus complete. 
\end{proof}

\section{$3$-wise Major Order Theorem with equality}

In the case when we are only interested in obtaining the relative ordering of a pair of alternatives in \emph{some} $3$-wise median (and thus not necessarily in every $3$-median), we can apply the following augmented version of the $3$-wise Major Order Theorem~\ref{3-wise-major-order}.

\begin{theorem}[\textbf{3MOTe}]
    \label{3-wise-major-order-with-equality}    Let $x,y$ be two alternatives of an election with voting profile $V$ over a set of alternatives $C$. Suppose that $\delta_{xy}\geq 0$ and 
     \begin{enumerate} [\rm (i)] 
     \item $Q_{x,y,z}\geq 0$ for all $z \in C\setminus \{x,y\}$, 
\item 
$2 \delta_{xy}   \geq \sum_{z\neq x,y} \max\left(0, - P_{x,y,z} + \sum_{t \neq x,y,z} \max\left(0, S_{y,x,z,t} \right) \right)$. 
%where the sum is over all alternatives $z,t \in C\setminus\{x,y\}$. 
\end{enumerate}
\par 
Then $x$ is ranked before $y$ in some $3$-wise median of the election. 
\end{theorem}

\begin{proof}
    The proof is the same, \emph{mutatis mutandis}, as the proof of Theorem~\ref{3-wise-major-order} where we notice that the inequality  $\Delta_1 + \Delta_2 \geq 0$ (cf.~\eqref{e:3-wise-major-order-eq-1}) is no longer strict but it is enough to conclude that either the ranking $\sigma_1^*$ or $\sigma_2^*$ is at least as good as $\pi$ as a $3$-wise consensus of the election. 
    The conclusion follows since in the rankings $\sigma_1^*$ and $\sigma_2^*$, we have $x>y$.  
\end{proof}

% \section{Kemeny schemes, committee elections,  incomplete rankings}

\section{Concluding remarks with  extended set-wise Kemeny schemes} 
\label{s:conclusion}
\par 
In real-life data, we can occasionally encounter incomplete rankings of alternatives, notably in surveys and political elections using for example the Single transferable vote method. These particular  situations correspond very well to the reasonable assumption  that typical voters only 
pay attention to the ranking of a  shortlist of their favorite alternatives and put a rather random order for the rest of the alternatives in complete votes or simply do not indicate any preferences on such alternatives. 
To deal with such a situation using the $3$-wise Kemeny voting scheme, we shall suppose that in an incomplete vote, all the alternatives which are not ranked will equally occupy the last position. The $3$-wise Kendall-tau distance can still apply for such incomplete rankings as well as the notion of the $3$-wise median, or more generally, $k$-wise median for every $k \geq 2$. 
The above method clearly increases the range of applicability of set-wise Kemeny voting schemes. See Appendix~\ref{appendix-examples} and Table~\ref{table:real-data} for some examples with real elections. 
We also observe that $3$-wise Kemeny voting scheme can be applied for committee elections where we would like to select a shortlist consisting of an arbitrarily fixed number of best alternatives according to certain criteria. 

\appendix

\section{$3$-wise Major Order Theorem as a measure of consensus}
\label{s:measure-consensus}

We can observe that in an election, the proportion of pairs whose relative order are determined by the $3$-wise Major Order Theorems (Theorem~\ref{3-wise-major-order} and Theorem~\ref{3-wise-major-order-iterated}) can be used as a good measure of the consensus level of the electorates. Hence, a high percentage of such pairs means that the electorates, as a whole, agree on the relative rankings of a high percentage of pairs of alternatives, since every $3$-wise median admits the same relative rankings on such pairs.  
We thus arrive as the following definition. 

\begin{definition}
Let $V$ be an election over $n$ alternatives and let $M$ be the number of pairs $(x,y)$ which satisfy  conditions (i) and (ii) in  Theorem \ref{3-wise-major-order}, and let $M^*$ be the number of pairs $(x,y)$ in $\cup_{k \geq 0} U_k$ in 
Theorem~\ref{3-wise-major-order-iterated}. Then the $\mathrm{3MOT}$-consensus level of the election $V$ is defined by: 
\begin{equation}
    \label{e:consensus-level-3-MOT}
    L_{\mathrm{3MOT}}(V) \coloneqq \frac{M}{\frac{n(n-1)}{2}} = \frac{2M}{n(n-1)}. 
\end{equation}
\par 
Similarly, the  $\mathrm{3MOT^*}$-consensus level of the election $V$ is defined by: 
\begin{equation}
    \label{e:consensus-level-3-MOT-iterated}
    L_{\mathrm{3MOT^*}}(V) \coloneqq \frac{M^*}{\frac{n(n-1)}{2}} = \frac{2M^*}{n(n-1)} \geq  L_{\mathrm{3MOT}}(V). 
\end{equation}
\end{definition}
\par 
In particular,  if the consensus levels $ L_{\mathrm{3MOT}}(V)$, $ L_{\mathrm{3MOT^*}}(V)$ of an election $V$ are high, then we have a lot of restrictions on the possibilities of $3$-wise medians and thus 
the total number of possible $3$-wise  medians is small. Clearly, a small number of $3$-wise medians of an election is a good indicator that the voters in the  electorate tend to agree with each other.  
See Table~\ref{table:real-data} for the consensus level $L_{\mathrm{3MOT}}$ and $L_{\mathrm{3MOT^*}}$ of several real-world elections recorded in PREFLIB \cite{prelib} and  Table~\ref{table:proportion-iterated-3mot-1},  Table~\ref{table:proportion-iterated-3mot-2} for the average consensus level of random elections for several  different values of the numbers of  alternatives and votes.

% \section{Comparison with the $3$-wise majority digraph method of Gilbert et al.} 

\section{$3$-wise scheme can be more resistant to manipulation}
\label{s:3-wise can be more resistant}
Note  that the unique $2$-wise median of the election $V$ in Example~\ref{ex:major-order-theorem-fail-3-wise} is $t>z>x>y$. Therefore, the unique $2$-wise winner is $t$ while arguably, the more convincing winner should be $z$ instead. 
Indeed, consider the following imaginary election over the same 4 alternatives $x,y,z,t$ representing 4 different political parties. Assume that the electorate consists of 11 million voters and the preferences of 10 million voters are fixed as follows. 
\begin{align*}
&r_1: z>t>x>y \,\,\text{ (5 million votes)},\quad \quad &r_2: y>t>x>z \,\,\text{ (2 million votes)}, \\
&r_3: x>y>t>z \,\,\text{ (0 vote)},\quad \quad &r'_3: x>y>z>t \,\,\text{ (1 million votes)},  \\
&r_4: t>x>y>z \,\,\text{ (2 million votes)}.\quad \quad & 
\end{align*}
\par 
\noindent 
Suppose that the remaining group $G$ of 1 million voters prefer $x$ and $y$ in that order and simply ignore the alternatives $z,t$. Hence, we can reasonably assume that among these voters, roughly half of them chose $t>z$ (see $r_3$) and the other half chose $z>t$ (see $r'_3$). 
Consider the following voting profile $V'$ where exactly half of the group $G$ prefer $z$ over $t$: 
\begin{align*}
&r_1: z>t>x>y \,\,\text{ (5 million votes)},\quad \quad & &r_2: y>t>x>z \,\,\text{ (2 million votes)}, \\
& r_3: x>y>t>z \,\,\text{ (0.5 million votes)},\quad \quad & &r'_3: x>y>z>t \,\,\text{ (1.5 million votes)},  \\
&r_4: t>x>y>z \,\,\text{ (2 million votes)}.\quad \quad &  &
\end{align*}

\par 
Then the only $2$-wise medians of the election $V'$ are $r_4$ and $r_1$. 
Moreover, if more than half of the group $G$ prefer $z$ over $t$ then $r_1$ is the unique $2$-wise median. Similarly, if more than half of the group $G$ prefer $t$ over $z$ then $r_4$ is the unique $2$-wise median. In other words, the group $G$ ultimately decides $z$ or $t$ by the majority rule as the unique $2$-wise winner. 
Consequently, if more than half of the group $G$ chose $x>y>t>z$, for example by strategic voting or under the influence of manipulative forces and even bribery of the political party of $t$, the ranking $r_4$ will be the unique $2$-wise median where $t$ surpasses $z$ and emerges as the unique $2$-wise winner.  
On the other hand, the ranking $r_1$ remains the unique $3$-wise median and thus $z$ is the indisputable $3$-wise winner of the election regardless of the preferences between $z$ and $t$ of voters in the group $G$ (by the computation with the voting profile $V$), even if all the voters in $G$ chose $x>y>t>z$.  
Therefore, we conclude that in such situations, the $3$-wise Kemeny voting scheme should be preferred over the classical Kemeny rule to avoid undesirable outcomes and manipulations.

\section{Simulations on real-world and uniform data and applications} 
\label{appendix-data}

\subsection{Examples with real-world data from PREFLIB}
\label{appendix-examples}

\begin{example} 
To illustrate, consider the 
2007 Glasgow City Council elections which were held by  Ward (cf. \cite[Dataset 00008]{prelib} and the data available on \hyperlink{https://en.wikipedia.org/w/index.php?title=2007_Glasgow_City_Council_election&oldid=1122885191}{Wikipedia}).   
Out of the total 21 Wards, the EastCentre Ward was allocated 4 seats in the city council which were   chosen using the Single transferable vote method among 13 alternatives by 9078 valid votes. The alternatives are:  
\begin{multicols}{3}
\begin{enumerate}[\rm 1)] 
    \item Jim Adams
\item Patricia Chalmers
\item Elaine Cooper
\item Drew Dickie
\item Frank Docherty
\item Jennifer Dunn
\item Stuart Grieve
\item David Johnston
\item John Kerr
\item Elaine Mcdougall
\item William Mclachlan
\item Daniel O'Donnell
\item Randle Wilson
\end{enumerate}
\end{multicols}
\par 
We shall identify each alternative with the order given above.  
Under the above described  extended $3$-wise Kemeny voting scheme, the $3$-wise Major Order Theorem \ref{3-wise-major-order} determines 67 pairs (out of 78 possible pairs) of the relative order of the form $(x,y)$, which means that $x>y$ in every $3$-wise median: 
\begin{align*}
 &  (1, 4), (1, 8), (1, 9), (1, 11), (1, 12), (1, 13), (2, 1), (2, 3), (2, 4), (2, 7), (2, 8), (2, 9),\\& (2, 11), (2, 12), (2, 13), (3, 4), (3, 8), (3, 9), (3, 11), (3, 12), (3, 13),(4, 8), (4, 11),\\ & (5, 1), (5, 3), (5, 4), (5, 7), (5, 8), (5, 9), (5, 10),  (5, 11), (5, 12), (5, 13), \\& (6, 1), (6, 3), (6, 4), (6, 7), (6, 8), (6, 9), (6, 11), (6, 12), (6, 13),\\& (7, 4), (7, 8), (7, 9), (7, 11), (7, 12), (7, 13), (9, 8), (9, 11),\\ &(10, 1), (10, 3), (10, 4), (10, 7), (10, 8), (10, 9), (10, 11), (10, 12), (10, 13),\\ & (12, 4), (12, 8), (12, 9), (12, 11), (13, 4), (13, 8), (13, 9), (13, 11). 
\end{align*}
\par 
Therefore, in every $3$-wise median, we have 
\begin{align*}
    5&>1,3,4,7,8,9,10, 11,12,13\\
   2, 6, 10& >1,3,4,7,8,9,11,12,13. 
\end{align*}
\par 
Hence, according to the $3$-wise Kemeny scheme, the alternatives 2,5,6,10, namely, Patricia Chalmers,  Frank Docherty, Jennifer Dunn, and Elaine Mcdougall, must win the Ward election. It turns out that the above $3$-wise election result coincides with the official result using the Single transferable vote method used in the 2007 Glasgow City Council elections.  
\par
Similarly, for the Shettleston Ward, there were 8803 valid  votes to determine 4 seats from 11 alternatives: 

\begin{multicols}{3}
\begin{enumerate}[\rm 1)] 
    \item  Mick Eyre
\item  Walter Hamilton
\item Wheatley Harris
\item David Jackson
\item Gordon Kirker
\item Catherine Maguire
\item Tom Mckeown
\item John F. Mclaughlin
\item Euan Mcleod
\item George Ryan
\item Alex White
\end{enumerate}
\end{multicols} 
\par 
The $3$-wise Major Order Theorem \ref{3-wise-major-order} determines the relative order of 45 pairs (out of 55 possible pairs) from which we deduce that:  
\begin{align*}
    7,8,9,10 &> 1,2,3,4,5,6,11,
\end{align*}
in every $3$-wise median. 
Hence, if we were to follow the $3$-wise Kemeny voting scheme, the 4 seats would belong to Tom Mckeown, John F. Mclaughlin, Euan Mcleod, and George Ryan, which is exactly the same official result of the election using the Single transferable vote method. 
\par 
The Newlands Ward has 3 seats to be determined from 8654 valid votes among 9 alternatives: 

\begin{multicols}{3}
\begin{enumerate}[\rm 1)] 
    \item  Kay Allan
\item Charlie Baillie
\item Eamonn Coyle
\item Stephen Curran
\item Colin Deans
\item Robert Mcelroy
\item Jim Mcnally
\item Gordon Morgan
\item Robert W Stewart
\end{enumerate}
\end{multicols} 
\par 
The $3$-wise Major Order Theorem \ref{3-wise-major-order} implies that 
 $4,5,7>1,2,3,6,7,8$ 
in every $3$-wise Kemeny. Hence, Stephen Curran, Colin Deans, and Jim Mcnally win  according to the $3$-wise Kemeny rule. Again, the same official result holds using the Single transferable vote method. 
\end{example} 
\par 
Our next example concerns one of the most important elections conducted with the Instant Runoff Voting (IRV) method.

\begin{example}
    Consider the UK Labor Party Leadership election in 2010  (see \cite[Dataset 00030]{prelib} and \href{https://rangevoting.org/LabourUK2010}{$\mathrm{rangevoting.org/LabourUK2010}$}) whose valid  electorate consists of 266 voters to vote for the following 5 alternatives identified with their number: 

    \begin{multicols}{3}
    \begin{enumerate} [\rm 1)]
        \item   Diane Abbott
\item Ed Balls
\item Andy Burnham
\item
David Miliband
\item
Ed Miliband.
    \end{enumerate}
     \end{multicols} 
  \par 
  While the number of pairs of the form $(x,y)$, which means that $x>y$ in every median, determined by the Unanimity property is zero, the Extended $2$-wise Always Theorem \cite{kien-sylvie-condorcet} applies and finds 6 such pairs out of 10 pairs in total, namely, $
(2, 1)$, $(3, 1)$, $(4, 1)$, $(5, 1)$, $(4, 3)$, $(5, 3)$. Hence, we deduce that 
$4,5>3>1$  and $2>1$ 
in every $2$-wise median of the UK Labor Leadership election in 2010. In fact, by  the $2$-wise Major Order Theorem (Theorem~\ref{2-wise-major-order}), the unique $2$-wise median of the election is $4>5>2>3>1$. 
\par 
With respect to the $3$-wise Kemeny rule, the $3$-wise Extended  Always Theorem \cite{kien-sylvie-condorcet} (Theorem~\ref{t:3-wise-unanimity-general})  determines the preferences of each pair (in the given order) $(2, 1)$, $(3, 1)$, $(4, 1)$, $(5, 1)$. Therefore, we have in every $3$-wise median that $
2,3,4,5>1$. 
\par 
On the other hand, the 3MOT (Theorem~\ref{3-wise-major-order}) determines 9 out of the total  10 pairs: 
\[
(2, 1), (2, 3), (3, 1), (4, 1), (4, 2), (4, 3), (5, 1), (5, 2), (5, 3)
\]
\par 
The remaining pair $(4, 5)$ is found by Iterated 3MOT (Theorem~\ref{3-wise-major-order-iterated}) after two iterations and thus the unique $3$-wise median of the election is given by 
\[
4>5>2>3>1  
\]
which coincides with the unique $2$-wise median. 
Consequently, the unique $3$-wise winner is David Miliband, who is also the winner with plain plurality (first-past-the-post), a method used in a preceding election in 1994. 
Note that the official winner (under the IRV method) is Ed Miliband, who won only by a small margin against David Miliband. 

\end{example}

\begin{example}

Table~\ref{table:real-data} describes the efficiency of the $3$-wise Major Order Theorems on real-world examples taken from PREFLIB of elections with a medium or large number of alternatives in various domains ranging from sport competitions, web search, and  university rankings to food and music preferences. The Python codes for Table~\ref{table:real-data} are available at https://github.com/XKPhung. 

\newpage 
\begin{table}[h!]
\centering
\caption{Proportion of pairs (out of all possible pairs) with relative order solved by the 3-wise Major Order Theorems \ref{3-wise-major-order} and \ref{3-wise-major-order-iterated} for real-world elections taken from PREFLIB \cite{prelib} where $n$ is the total  number of alternatives, $m$ is the total number of votes, and $n(n-1)/2$ is the total number of pairs of alternatives. In the column Iterated 3MOT,  ($k$ iter.) signifies that Iterated 3MOT does not generate new information starting from the $(k+1)$-th iteration. The symbol $\geq$ in the column Iterated 3MOT means that, due to lack of time, we did not run  Iterated 3MOT on the corresponding instance. However, the result must be at least as good as 3MOT. The last column gives a lower bound on the reduction rate of the search space by Iterated 3MOT (cf. \eqref{e:bound-restriction-space}).  
Thus, the size of the search space for 3-wise medians is bounded by $n! / (\text{reduction rate})$.} 
 
\label{table:real-data}
% \begin{longtable}[c]{|m{6cm}|m{2cm}|m{2cm}|}
% \begin{tabulary} {\linewidth} {C C C}
\begin{tabularx}{\textwidth}{ 
  | >{\hsize=1.4\hsize \centering\arraybackslash}X 
  | >{\hsize=0.4\hsize \centering\arraybackslash}X 
  | >{\hsize=0.4\hsize \centering\arraybackslash}X 
   | >{\hsize=0.8\hsize \centering\arraybackslash}X 
  | >{\hsize=0.7\hsize\centering\arraybackslash}X 
 | >{\hsize=1.2\hsize\centering\arraybackslash}X 
  | >{\hsize=1.0\hsize\centering\arraybackslash}X |}
% \begin{tabular}{ |m{4cm}|m{4cm}|m{4cm}| }
%\begin{tabular}{ *{3}{c}}
%\hline
%\multicolumn{5}{|c|}{merge column} \\
\hline
\textbf{Dataset number} & \textbf{n} & \textbf{m}  & \textbf{n(n - 1)/2} & \textbf{3MOT}  &\textbf{Iterated 3MOT} & \textbf{Reduction rate $\geq$}\\ 
\hline  \hline   
  00056-00003990.soc
  & 256 & 17 &  32640 & 60.82$\%$ &  $\geq$ 60.82$\%$ & $2.94\times 10^{69}$
\\  
\hline   
 00015-00000002.soc  & 240 & 5 & 28680 & 56.86$\%$ &  $\geq$ 56.86$\%$     & $1.74\times 10^{55}$
%  \\ 
%  \hline   
% 00046-00000003.soc  & 200 & 19 &  19900 & 19.05$\%$ &  $\geq$ 19.05$\%$   &
\\ 
\hline   
 00015-00000023.soc  & 142 & 4 & 10011 & 45.71$\%$ &  $\geq$ 45.71$\%$  & $4.39\times 10^{18}$
\\   
\hline   
 00056-00003986.soc   & 130 & 5 &  8385 & 87.67$\%$ &  $\geq$ 87.67$\%$   & $4.66\times 10^{95}$
 \\ 
\hline   
00015-00000033.soc  & 128 & 4 & 8128 & 45.14$\%$ &  $\geq$ 45.14$\%$   & $1.62\times 10^{16}$
 \\ 
\hline   
00015-00000018.soc  & 115 & 4 &  6555 & 46.90$\%$ &  48.88$\%$ (3 iter.) & $8.82\times 10^{17}$
\\   
\hline   
00014-00000002.soi  & 100 & 5000 &  4950 &  85.17$\%$ &  $\geq$ 85.17$\%$     & $2.91\times 10^{65}$
\\   
\hline   
00015-00000069.soc  & 81 & 4 & 3240 & 52.65$\%$ &  59.14$\%$ (10 iter.)  &  $9.55\times 10^{19}$ \\ 
\hline 
00049-00000010.soi & 64 & 11 & 2016 & 72.47$\%$ & 93.80$\%$ (19 iter.)  & $5.20\times 10^{58}$
\\
\hline 
00049-00000049.soi & 55 & 27 & 1485 & 40.60$\%$ & 43.43$\%$ (4 iter.)   & $1.15 \times 10^6$
\\
\hline 
00049-000000114.soi & 40 & 40 & 780 & 64.49$\%$ & 73.97$\%$ (4 iter.)   & $1.58 \times 10^{16}$
\\
\hline 
00049-00000083.soi & 38 & 38 & 703 & 57.75$\%$ & 64.01$\%$ (4 iter.)   & $5.48 \times 10^{10}$
\\
\hline 
00049-00000116.soi & 29 & 45 & 406 & 64.29$\%$ & 70.44$\%$ (3 iter.) & $9.45 \times 10^{9}$ \\ 
\hline 
00006-00000028.soc & 24 &  9 & 276 & 89.86$\%$ & 97.46$\%$ (3 iter.)  & $1.05 \times 10^{21}$
\\ 
\hline
00052-00000070.soc & 20 & 21 & 190 & 62.63$\%$ & 94.74$\%$ (5 iter.)   & $5.19 \times 10^{14}$
\\
\hline 
00053-00000362.soi & 19 & 68 & 171 & 78.36$\%$ & 96.49$\%$ (4 iter.)& $5.13 \times 10^{14}$
\\ 
\hline 
 00053-00000383.soi  & 17 & 50 &  136 & 67.65$\%$ &  92.65$\%$ (4 iter.) & $9.28 \times 10^{10}$\\   
\hline  
 00035-00000002.soc & 15 & 42 & 105 & 50.48$\%$ & 68.57$\%$ (4 iter.) & $1.69 \times 10^4$
 \\
\hline 
00035-00000003.soc & 15 & 42 & 105 & 71.43$\%$ & 90.48$\%$ (4 iter.) & $4.04\times 10^8$
\\ 
\hline 
00035-00000004.soc & 15 & 42 & 105 & 41.90$\%$ & 79.05$\%$ (5 iter.) & $9.21\times 10^5$
\\
\hline 
00035-00000005.soc & 15 & 42 & 105 & 35.24$\%$ & 72.38$\%$ (5 iter.) & $6.40 \times 10^4$
\\ 
\hline
00035-00000006.soc & 15 & 42 & 105 & 48.57$\%$ & 88.57$\%$ (5 iter.) & $1.21 \times 10^8$
\\ 
 \hline 
 00035-00000007.soc & 15 & 42 & 105 & 66.67$\%$ & 91.43$\%$ (4 iter.) & $7.64 \times 10^8$
 \\ 
\hline 
00001-00000003.soi & 14 & 6481 & 91 & 79.12$\%$ & 92.31$\%$ (2 iter.) & $2.10 \times 10^8$
\\
% \hline 
% 00012-00000001.soc & 11 & 30 & 55 & 67.27$\%$ & 85.45$\%$ (3 iter.) & $6.20 \times 10^4$
% \\ 
\hline
 
%\end{tabular}
% \end{tabulary}
\end{tabularx}
% \end{longtable} 
\end{table}

\end{example}

\subsection{Tests on uniform data} 

% e.g. Mallows model, using the Python package PrefLib-Tool as in the paper of Gilbert et al. 
% % \par 
For the hard case of uniformly generated data,  we shall compare the performance of our $3$-wise Major Order Theorems with the following two theoretical space reduction results recently found in \cite{kien-sylvie-condorcet}. The first one is the $3$-wise Extended Always theorem 
 which generalizes the unanimity property \cite[Proposition 5]{setwise} also known as the Pareto efficiency and the second one is the $5/6$-majority rule which is the $3$-wise counterpart of the well-known $3/4$-majority rule of Betzler et al. \cite{3/4}. 

\begin{theorem}[\textbf{3AT}, \cite{kien-sylvie-condorcet}] 
\label{t:3-wise-unanimity-general}
Let $x,y$ be candidates in an election with $n \geq 2$ candidates. Suppose that $x \geq_\alpha y$ for some $\alpha \in [0,1] $ such that 
\begin{equation*}
 \alpha > g(n)= 1 -  \frac{1}{n^2-3n+4}. 
\end{equation*}
 \par 
 Then $x>y$ in every $3$-wise median of the election. 
\end{theorem}

\begin{theorem}[\textbf{5/6-majority rule}, \cite{kien-sylvie-condorcet}]
\label{l:5/6-majority-3-wise-main}
Let 
$x$ be a non-dirty candidate in an election 
 with respect to the $5/6$-majority rule and let $I= \{z \neq x \colon z\geq_{5/6}x\}$. 
Let $r$ be a $3$-wise median  such that $z>x$ in $r$ for all $z \in I$. Suppose that  $|I|(|I|-4) \leq 3 |\{ z \colon x> z \text{ in } r\}|$. 
Then for every candidate  $y\neq x$ with $x \geq_{5/6} y$, we have $x > y$  in $r$. 
\end{theorem}
\par 
\noindent 
The simulation results are given Table~\ref{table:proportion-iterated-3mot-1} and Table~\ref{table:proportion-iterated-3mot-2} below. Observe that in general, the performance of the $3$-wise Major Order Theorems is better when the number $m$ of votes is odd than when the number of votes (with a comparable size) is even. This phenomenon is explained by the property that unlike the case when $m$ is odd, the probability that $\delta_{xy} =0$ is strictly higher than 0 when $m$ is even, which reduces the applicability of the $3$-wise Major Order Theorems. Moreover, as for the classical $3/4$-majority rule of Betzler et al. \cite{3/4}, the applicability of the $3$-wise $5/6$ majority rule drops quickly to almost $0\%$ on average for elections with at least $10$ alternatives and at least $15$ votes. Hence, for ease of reading, we only report the simulation results for the $3$-wise Extended Always Theorem \cite{kien-sylvie-condorcet} and the 3-wise Major Order Theorems \ref{3-wise-major-order},  \ref{3-wise-major-order-iterated}, and \ref{3-wise-major-order-with-equality}.  
The Python codes for Table~\ref{table:proportion-iterated-3mot-1}, Table~\ref{table:proportion-iterated-3mot-2}  are available  at https://github.com/XKPhung. Note also that for each pair  $(n,m)$ with $3\leq n,m \leq 30$, the approximate average maximal number of iterations of Iterated 3MOT ranges from 2 to 6. For Kemeny's rule, due to lack of time, we do not provide similar complete comparison between the performance of Theorem~\ref{t:improved-2MOT-3.0} and Theorem~\ref{iterated-2-wise-major-order} on uniformly generated data. However, by realizing several quick simulations, we find that for $n,m \leq 30$, Theorem~\ref{t:improved-2MOT-3.0} can solve approximately from 2$\%$ to 7$\%$ more pairs than Theorem~\ref{iterated-2-wise-major-order}.

\newpage

\begin{table}[h!]
\centering
\caption{Proportion of pairs with relative order solved by the $3$-wise  Extended Always Theorem \cite{kien-sylvie-condorcet} and the 3-wise Major Order Theorems \ref{3-wise-major-order},  \ref{3-wise-major-order-iterated}, and \ref{3-wise-major-order-with-equality} over 100 000 uniformly generated random instances where $n$ is the total  number of alternatives and $m$ is the total number of votes.}

% Due to lack of time and resources and since an approximation of the efficiency is sufficient for our purpose, the instances of voting profiles for different columns (algorithms) may differ in general but are all uniformly generated.}
% \begin{longtable}[c]{|m{4cm}|m{4cm}|m{4cm}|}
\label{table:proportion-iterated-3mot-1}
% \begin{tabulary} {\linewidth} {C C C}
\begin{tabularx}{1\textwidth}{ 
  | >{\centering\arraybackslash}X 
  | >{\centering\arraybackslash}X 
  | >{\centering\arraybackslash}X 
  | >{\centering\arraybackslash}X 
  | >{\centering\arraybackslash}X 
  | >{\centering\arraybackslash}X |}
% \begin{tabular}{ |m{4cm}|m{4cm}|m{4cm}| }
%\begin{tabular}{ *{3}{c}}
%\hline
%\multicolumn{5}{|c|}{merge column} \\
\hline
\textbf{n} & \textbf{m} & \textbf{3AT}  &\textbf{3MOT} & \textbf{3MOTe} & \textbf{Iterated 3MOT}\\ 
\hline  \hline 
3 & 3 &   25.11\% &  86.14\% & 94.50\% &  94.46\%  
\\ 
\hline  
- & 4 &   12.44\% &  61.39\%  & 98.30\%  &64.48\%
\\ 
\hline  
 - & 5 &    37.57\% &  86.16\%  & 86.16\%  &94.63\%
 \\ 
\hline
 - & 6 &    21.76\% &  65.42\%  & 97.56\%  & 68.55\%
 \\ 
\hline
 - & 7 &    12.51\% & 83.17\%  & 86.05\% & 93.75\%
 \\ 
\hline
 - & 8 &  7.00\% & 68.10\%  & 93.36\%  & 71.76\%
 \\ 
\hline
 - & 9 &  18.03\% &  81.91\%  & 84.82\% & 91.90\%
 \\ 
\hline
- & 10 &  11.04\% & 68.79\%  & 92.30\%  &73.15\%
 \\ 
\hline
- & 15 &   3.51\% &  81.13\%  & 82.85\%   &90.92\%
 \\ 
\hline\hline 
4 & 3 &   24.93\% &  75.62\% & 81.44\% & 89.63\%
 \\ 
\hline
- & 4 &    12.44\% & 59.35\%  & 80.53\% &64.42\%
 \\ 
\hline
- & 5 &    6.26\% & 72.59\%   & 75.12\%  &88.37\%
 \\ 
\hline
- & 6 &   3.10\% &  60.55\%  &  78.52\% & 65.75\%
 \\ 
\hline
- & 7 &    1.52\% &  70.10\%  & 73.96\% &86.20\%
 \\ 
\hline
- & 8 &    0.78\% & 61.55\%   & 75.45\%  &68.14\%
 \\ 
 \hline 
 - & 9 &    3.87\% &  69.40\%  & 72.11\%  & 84.19\%
 \\ 
\hline
- & 10 & 2.12\% &  61.75\%  & 74.95\% &68.99\%
 \\ 
\hline
- & 15 &    0.10\% &   68.02\%  & 69.03\%  & 82.21\%
 \\ 
\hline\hline 
5 & 3 &    25.05\% &  66.46\%  & 71.34\% & 85.20\%
 \\ 
\hline
- & 4 &    12.52\% &   54.90\% & 68.96\% & 60.90\%
 \\ 
\hline
- & 5 &   6.26\% &   62.58\% & 65.63\%  & 81.97\%
 \\ 
\hline
- & 6 &    3.15\% & 55.43\%  & 66.53\% & 62.12\%
 \\ 
 \hline
 - & 7 &    1.59\% &   60.46\%   & 64.13\% &78.60\%
 \\ 
\hline
- & 8 &   0.80\% & 55.34\%   & 64.00\% &63.77\%
 \\ 
\hline
- & 9 &    0.40\% &  59.85\%  & 62.29\% &76.85\%
 \\ 
\hline
- & 10 &    0.21\% &   55.23\% & 63.52\%  &64.17\%
 \\ 
\hline
- & 15 &    0.10\% &  58.57\%  & 60.35\%  &74.49\%
 \\ 
\hline
%\end{tabular}
% \end{tabulary}
\end{tabularx}
% \end{longtable} 
\end{table}

\newpage 
 
\begin{table}[h!]
\centering
\caption{Proportion of pairs with relative order solved by the $3$-wise Extended Always Theorem \cite{kien-sylvie-condorcet} and the 3-wise Major Order Theorems \ref{3-wise-major-order},  \ref{3-wise-major-order-iterated}, and \ref{3-wise-major-order-with-equality} over 100 000 random  instances where $n$ is the total  number of alternatives and $m$ is the total number of votes. For 3MOT and 3MOTe with $n=15$, the simulation is realized with random 10 000 instances. For $n=8$, resp. $n\in \{10,15\}$, the simulation for Iterated 3MOT is realized with random 10 000, resp. 1000,  instances. For $n=20$, all the simulations are realized with random 1000 instances.  All instances are uniformly generated. 
% Due to lack of time and resources and since an approximation of the efficiency is sufficient for our purpose, the instances of voting profiles for different columns (algorithms) may differ in general but are all uniformly generated.
}
\label{table:proportion-iterated-3mot-2}
% \begin{longtable}[c]{|m{4cm}|m{4cm}|m{4cm}|}

% \begin{tabulary} {\linewidth} {C C C}
\begin{tabularx}{1\textwidth}{ 
  | >{\centering\arraybackslash}X 
  | >{\centering\arraybackslash}X 
  | >{\centering\arraybackslash}X 
  | >{\centering\arraybackslash}X 
  | >{\centering\arraybackslash}X 
  | >{\centering\arraybackslash}X |}
% \begin{tabular}{ |m{4cm}|m{4cm}|m{4cm}| }
%\begin{tabular}{ *{3}{c}}
%\hline
%\multicolumn{5}{|c|}{merge column} \\
\hline
\textbf{n} & \textbf{m} & \textbf{3AT} & \textbf{3MOT} & \textbf{3MOTe} &\textbf{Iterated 3MOT}\\ 
\hline  \hline 
8 & 3 &   24.93\% & 48.97\% & 52.10\% &  72.12\%
\\ 
\hline  
- & 4 &   12.52\%  & 43.22\%  & 48.86\%  &50.57\%
\\ 
\hline  
 - & 5 &    6.25\% & 43.97\%  &  46.71\% &64.33\%
 \\ 
\hline
 - & 6 &    3.12\% & 41.71\%  &  46.00\% & 50.82\%
 \\ 
\hline
 - & 7 &    1.56\% & 42.45\%  &  44.93\% &59.95\%
 \\ 
\hline
 - & 8 &  0.79\% & 40.91\%  &  44.35\% &51.15\%
 \\ 
\hline
 - & 9 &  0.39\% & 41.94\%  & 43.76\%  &58.32\%
 \\ 
\hline
- & 10 &  0.19\% & 40.38\%  &  43.57\% &50.83\%  
 \\ 
\hline
- & 15 &   0.01\% & 40.57\%  & 41.92\%  &55.65\%
 \\ 
\hline\hline 
10 & 3 &   24.97\% &  42.83\% & 45.14\% &65.00\%
 \\ 
\hline
- & 4 &  12.51\% & 37.37\%  & 41.14\% &45.24\%
 \\ 
\hline
- & 5 &    6.23\% &36.64\%   & 38.85\% &54.17\%
 \\ 
\hline
- & 6 &  3.13\%  & 35.14\%  & 38.01\% &41.44\%
 \\ 
\hline
- & 7 & 1.55\% & 35.21\%  & 37.17\%   &50.60\%
 \\ 
\hline
- & 8 &  0.79\%    & 34.18\%   & 36.52\%  & 44.02\%
 \\ 
 \hline
 - & 9 &  0.39\%  & 34.59\%  &  36.08\% &50.07 \%
 \\ 
\hline
- & 10 & 0.20\% & 33.62\%  & 35.75\%  &43.48\%
 \\ 
\hline
- & 15 &    0.01\% &  33.22\%  &  34.34\%  &46.74\%
 \\ 
\hline\hline 
15 & 3 &  24.99\% & 35.20\%  & 36.43\%  & 51.36\%
 \\ 
\hline
- & 4 &    12.49\% & 28.14\% & 30.01\%  & 35.45\% 
 \\ 
\hline
- & 5 &   6.27\% &  25.92\% & 27.32\%  & 38.54\%
 \\ 
\hline
- & 6 &    3.12\% & 24.72\%  & 26.22\% & 32.62\%
 \\ 
 \hline
 - & 7 &    1.57\% & 24.19\%   &  25.41\% & 34.56\%
 \\ 
\hline
- & 8 &    0.78\% &  23.76\%   & 25.01\% & 31.90\% 
 \\ 
\hline
- & 9 &    0.39\% &  23.5\%  & 24.51\%  & 33.16\%
 \\ 
\hline
- & 10 &    0.20\% &   23.17\% & 24.28\% & 30.82\% 
 \\ 
\hline
- & 15 &     0.01\% & 22.25\%  & 22.98\%  & 30.91\% 
 \\ 
\hline \hline 
20 & 3 & 25.38\%  & 32.07\% & 32.81\%  & 42.90\% 
\\\hline 
- & 4 & 12.40\% & 23.04\% & 24.13\%  & 28.65\%
\\\hline 
- & 5 & 6.06\% & 19.77\% &20.74\%  & 28.11\%
\\\hline 
- & 6 & 3.16\% & 18.89\%  &19.88\% & 25.16\%
\\\hline 
- & 7 & 1.56\% & 18.33\%  &19.16\% & 25.68\% 
\\\hline 
- & 8 & 0.77\% & 17.69\% & 18.50\% & 23.76\% 
\\\hline 
- & 9 & 0.39\% & 17.14\% & 17.87\% & 24.27\%
\\\hline 
- & 10 & 0.19\% & 17.12\% & 17.87\% & 23.30\% 
\\\hline 
- & 15 & 0.01\% & 16.29\% & 16.82\%  &  22.44\% 
\\\hline 

%\end{tabular}
% \end{tabulary}
\end{tabularx}
% \end{longtable} 
\end{table}

\clearpage

\bibliography{Kemeny}

\end{document}